\documentclass[journal]{IEEEtran}
\usepackage{bbm}
\usepackage{array}
\usepackage{dcolumn}
\usepackage{epsfig}
\usepackage[intlimits]{amsmath}
\usepackage{amsmath, amsfonts, yhmath, bm}
\usepackage{amssymb}
\usepackage{psfrag}
\usepackage{color,soul}
\usepackage{xcolor,cancel}
\usepackage[normalem]{ulem}
\usepackage{enumerate}
\usepackage{stackengine}
\usepackage[noadjust]{cite}
\usepackage{graphicx}
\usepackage{caption}
\usepackage{subcaption}
\usepackage{cite}
\usepackage[font=footnotesize]{caption}

\newtheorem{thm}{Theorem}
\newtheorem{lem}{Lemma}
\newtheorem{proof}{Proof}

\definecolor{purple}{rgb}{0.6,0.0,0.7}

\definecolor{cyan}{rgb}{0.2,0.6,0.7}

\newcommand {\myvec}[1] {{\mbox{\boldmath $#1$}}}
\newcommand {\mymat}[1]  {{\mbox{\boldmath $#1$}}}

\newcommand {\mS} {\mymat{S}}
\newcommand {\mX} {\mymat{X}}
\newcommand {\A} {\mymat{A}}
\newcommand {\T} {\mymat{T}}
\newcommand {\hT} {\widehat{\T}}
\newcommand {\meps} {\mymat{\varepsilon}}
\newcommand {\Omeg} {\mymat{\Omega}}

\newcommand {\mBeta} {\mymat{\beta}}
\newcommand {\mLambda} {\mymat{\Lambda}}

\newcommand {\bA} {\mybar{\A}}
\newcommand {\bB} {\mybar{\B}}

\newcommand {\bC} {\mybar{\C}}
\newcommand {\bP} {\mybar{\P}}
\newcommand {\B} {\mymat{B}}

\newcommand {\hB} {\widehat{\B}}

\newcommand {\hQ} {\widehat{\Q}}

\newcommand {\C} {\mymat{C}}
\newcommand {\D} {\mymat{D}}

\newcommand {\mPhi} {\mymat{\Phi}}

\newcommand {\E} {\mymat{E}}

\newcommand {\G} {\mymat{G}}
\newcommand {\hG} {\widehat{\G}}
\renewcommand {\H} {\mymat{H}}

\renewcommand {\O} {\textrm{O}}
\renewcommand {\P} {\mymat{P}}

\newcommand {\tP} {\widetilde{\P}}
\newcommand {\tbP} {\widetilde{\bP}}

\newcommand {\Q} {\mymat{Q}}

\newcommand {\tQ} {\widetilde{\Q}}
\newcommand {\htQ} {\widehat{\tQ}}

\newcommand {\wtC} {\widetilde{\C}}
\newcommand {\wtP} {\widetilde{\P}}

\newcommand {\I} {\mymat{I}}
\newcommand {\X} {\mymat{X}}
\newcommand {\Y} {\mymat{Y}}
\newcommand {\F} {\mymat{F}}
\newcommand {\ue} {\myvec{e}}
\newcommand {\ua} {\myvec{a}}
\newcommand {\ub} {\myvec{b}}

\newcommand {\uepsilon} {\myvec{\varepsilon}}
\newcommand {\tb} {\tilde{\ub}}

\newcommand {\ug} {\myvec{g}}

\newcommand {\uq} {\myvec{q}}
\newcommand {\huq} {\widehat{\uq}}
\newcommand {\uu} {\myvec{u}}
\newcommand {\ualpha} {\myvec{\alpha}}

\newcommand {\uo} {\myvec{0}}
\newcommand {\us} {\myvec{s}}

\newcommand {\bs} {\bar{\us}}

\newcommand {\uy} {\myvec{y}}

\newcommand {\uh} {\myvec{h}}
\newcommand {\utheta} {\myvec{\theta}}

\newcommand {\Rset} {\mathbb{R}}

\newcommand {\Zset} {\mathbb{Z}}
\newcommand {\Tr} {\text{Tr}}
\newcommand {\tps} {\rm{T}}

\makeatletter
\newsavebox\myboxA
\newsavebox\myboxB
\newlength\mylenA

\newcommand*\mybar[2][0.75]{%
    \sbox{\myboxA}{$\m@th#2$}%
    \setbox\myboxB\null
    \ht\myboxB=\ht\myboxA%
    \dp\myboxB=\dp\myboxA%
    \wd\myboxB=#1\wd\myboxA
    \sbox\myboxB{$\m@th\overline{\copy\myboxB}$}
    \setlength\mylenA{\the\wd\myboxA}
    \addtolength\mylenA{-\the\wd\myboxB}%
    \ifdim\wd\myboxB<\wd\myboxA%
       \rlap{\hskip 0.5\mylenA\usebox\myboxB}{\usebox\myboxA}%
    \else
        \hskip -0.5\mylenA\rlap{\usebox\myboxA}{\hskip 0.5\mylenA\usebox\myboxB}%
    \fi}
\makeatother

\def\vec{\mathop{\rm vec}\nolimits}

\def\comment#1{}

\DeclareSymbolFont{grb}{OML}{cmm}{b}{it}
\DeclareMathSymbol{\lambdab}{\mathord}{grb}{"15}

\newcommand{\stkout}[1]{
	\color{red}\ifmmode\text{\sout{\ensuremath{#1}}}\else\sout{#1}\fi\color{black}}

\newcommand{\addra}{}
\newcommand{\delra}{\comment}
\newcommand{\delralign}{\comment}

\begin{document}

\title{Performance Analysis of the Gaussian Quasi-Maximum Likelihood Approach for Independent Vector Analysis}

\author{Amir Weiss$^1$, Sher Ali Cheema$^2$, Martin Haardt$^2$, and Arie Yeredor$^1$

\thanks{$^1$ Authors with Tel-Aviv University, School of Electrical Engineering,
P.O.~Box 39040, Tel-Aviv 69978, Israel, e-mail:
amirwei2@mail.tau.ac.il, arie@eng.tau.ac.il }
\thanks{
$^2$ Authors with
Ilmenau University of Technology,
Communications Research Laboratory,
P.~O.~Box 10~05~65, D-98684 Ilmenau, Germany,
e-mail: \{sher-ali.cheema, martin.haardt\}@tu-ilmenau.de,
phone: +49 (3677) 69-2613,
fax: +49 (3677) 69-1195,
WWW: http://www.tu-ilmenau.de/crl.}
}

\maketitle

\begin{abstract}
Maximum Likelihood (ML) estimation requires precise knowledge of the underlying statistical model. In Quasi ML (QML), a presumed model is used as a substitute to the (unknown) true model. In the context of Independent Vector Analysis (IVA), we consider the Gaussian QML \delra{e}\addra{E}stimate\addra{ (QMLE)} of the demixing matrices set and present an (approximate) analysis of its asymptotic separation performance. In Gaussian QML the sources are presumed to be Gaussian, with covariance matrices specified by some ``educated guess". The resulting quasi-likelihood equations of the demixing matrices take a special form, recently termed an extended ``Sequentially Drilled" Joint Congruence (SeDJoCo) transformation, which is reminiscent of (though essentially different from) classical joint diagonalization. We show that asymptotically this QML\addra{E}\delra{ estimate}, i.e., the solution of the resulting extended SeDJoCo transformation, attains perfect separation (under some mild conditions) regardless of the sources' true distributions and/or covariance matrices. In addition, based on the ``small-errors" assumption, we present a first-order perturbation analysis of the extended SeDJoCo solution. Using the resulting closed-form expressions for the errors in the solution matrices, we provide closed-form expressions for the resulting Interference-to-Source Ratios (ISRs) for IVA. Moreover, we prove that asymptotically the ISRs depend only on the sources' covariances, and not on their specific distributions. \addra{As an immediate consequence of this result, we provide an asymptotically attainable lower bound on the resulting ISRs. }We also present empirical results, corroborating our analytical derivations, of three simulation experiments concerning two possible model errors - inaccurate covariance matrices and sources' distribution mismodeling.
\end{abstract}



\section{Introduction}
\label{sec_Intro}
The Blind Source Separation (BSS) problem \cite{cardoso1990eigen,jutten1991blind,comon1991blind,comon2010handbook} consists of retrieving signals of interest, termed the sources, from a single dataset consisting of their mixtures. One of the most popular and common paradigms for solving the BSS problem is Independent Component Analysis (ICA) \cite{comon2010handbook,comon1992independent,hyvarinen2004independent}, where the sources are assumed to be (only) mutually statistically independent random processes, and the mixtures are assumed to be linear combinations thereof, where the linear mixing operator is unknown.

Joint BSS (JBSS) \cite{li2009joint,li2011joint,anderson2012joint} is an extension of the BSS problem, where multiple datasets of mixtures are observed. JBSS is commonly addressed under the Independent Vector Analysis (IVA) paradigm \cite{kim2006independent,anderson2014independent,lee2007fast}, a (natural) extension of ICA, where each dataset is restricted to the ICA formulation, with the addition of a potentially allowed statistical dependence between each source in one dataset and (at most) one source in every \textit{different} dataset. The interest in IVA emerged due to the nature of problems with dependence between multiple datasets such as analysis of multi-subject fMRI data \cite{lee2008independent,dea2011iva,engberg2016independent} or the convolutive ICA problem, formulated in the frequency domain, using multiple frequency bins \cite{kim2006independent2,kim2010real}.

ICA and IVA with their conventional assumptions (mentioned above) describe a fully blind scenario, i.e., where no additional information is available. However, in cases where some {\it a-priori} knowledge is available, be it full/partial description of the sources' statistics or some information regarding the (linear) mixing operator(s), the scenario is termed ``semi-blind" \cite{gunther2012learning,hesse2006semi,nesta2011batch}. A particularly interesting case is when the sources' probability distributions are known {\it a-priori}, allowing the Maximum Likelihood (ML) approach to be taken. ML separation is attractive due to its (asymptotic) optimality \cite{degerine2004separation,doron2007cramer,yeredor2010blind,weiss2017the} in the sense of minimal attainable Interference to Source Ratio (ISR), a common measure for the separation performance. Concentrating on the IVA problem (of which ICA is merely a particular case) with an equal number of mixtures and sources in every dataset, it has been shown in \cite{weiss2017the,cheng2015extension}, that when the sources are zero-mean Gaussian with known and distinct temporal covariance matrices, the ML \delra{e}\addra{E}stimate\addra{ (MLE)} of the demixing operators, which in this case are essentially the inverses of the mixing matrices, may be obtained by the solution of the so called extended ``Sequentially Drilled" Joint Congruence (SeDJoCo) transformation (previously presented, though not with this name, for a single dataset (ICA) by Pham and Garat in \cite{pham1997blind}, D{\'e}gerine and Za{\"\i}di in \cite{degerine2004separation,degerine2006determinant} and Yeredor {\textit{et al.}} in \cite{yeredor2009hybrid,yeredor2010blind,yeredor2012sequentially}), which are obviously the likelihood equations in this context. In addition, a few interesting properties of the solution and two iterative solution algorithms for this problem have been proposed as well in \cite{weiss2017the}, as we shall elaborate in Section II. It is interesting to remark, that in the context of Multi-User Multiple-Input Multiple-Output (MU-MIMO) Coordinate Beamforming (CBF) setup (see \cite{cheng2015extension,cheng2016extension} for further details), the solution of (an almost identical problem to) the extended SeDJoCo equations serves as the beamforming transformation which achieves the elimination of multiuser interference as well as the maximization of the desired signal components.

On the other hand, Quasi Maximum Likelihood (QML) approaches (e.g., \cite{pham2001blind,ghogho2000,todros2010,bronstein2005quasi}), make some model assumptions on the sources, and use some ``educated guess" for the associated parameters, in order to facilitate a ``quasi-" ML\addra{E}\delra{ estimate}, which would hopefully approximate the ML\addra{E}\delra{ estimate} when the assumed model is close to reality. \addra{Thus, a QML \delra{e}\addra{E}stimate\addra{ (QMLE)} is defined as any estimate which can be interpreted as the ML\addra{E}\delra{ estimate} under an assumption of some presumed, hypothesized model, not necessarily describing the true state of nature.} For example, if the sources are presumed to be Gaussian with known temporal auto- and cross-covariances (between respective sources across different datasets only, as in a standard IVA setup), the implied likelihood of the observed mixtures from all datasets is expressed and maximized (with respect to the unknown mixing matrices), essentially resulting in a set of extended SeDJoCo equations. However, since the sources are not necessarily Gaussian, and their temporal auto- and cross-covariances are actually unknown, in this case the resulting extended SeDJoCo equations are in fact the \textit{quasi}-likelihood (rather than the likelihood) equations.

While the true ML\addra{E}\delra{ estimate} enjoys some appealing, well-known properties, such as consistency and asymptotic efficiency \cite{vantrees2013detection}, these properties are generally not shared by QML\addra{E}\delra{ estimate}s. Nevertheless, in this paper we show that the \addra{Gaussian }QML\addra{E}\delra{ estimate}s of the demixing matrices, obtained by a solution to an extended SeDJoCo transformation for the general IVA problem, are consistent\footnote{An estimate is considered consistent in the context of IVA/ICA if its resulting ISRs all tend to zero (perfect separation) when the observation lengths tend to infinity.} estimates of the separating matrices (under some mild conditions), i.e., their consistency is not restricted to a semi-blind scenario. Namely, these estimates are consistent even when the sources' temporal covariance matrices are unknown and/or when the sources are not Gaussian. Consequently, the solution of the extended SeDJoCo equations is apparently of considerable interest in the general context of IVA. Moreover, in reality, even in the semi-blind scenario the \textit{a-priori} information (or assumptions) are not likely to be exact, e.g., error-prone estimated versions of the covariance matrices and/or an approximate distribution of the sources might be available. Consequently, the performance analysis of the extended SeDJoCo solution is of high interest also in various (more realistic) scenarios, as detailed in the sequel, establishing our motivation for this work.
Our main contributions in this paper are as follows:
\begin{itemize}
	\item A ``\delra{F}\addra{f}irst-order" perturbation analysis of the extended SeDJoCo solution: We provide a full analytical derivation with closed-form expressions for the errors in the solution matrices resulting from perturbations in the coefficients under the ``small-errors" assumption (i.e., when neglecting second- and higher-order terms). This result is not confined to the context of JBSS and may be used in other contexts as well (e.g., in the MU-MIMO CBF problem mentioned earlier\footnote{with very slight modifications}).
	\item Consistency of the extended SeDJoCo solution as the QML\addra{E}\delra{ estimate}: We show that a solution of extended SeDJoCo, i.e., a solution to the quasi-likelihood equations, provides a consistent estimate of the demixing matrices in the IVA scenario (under some mild conditions), even if the {\it a-priori} information is unavailable or inaccurate.
	\item Asymptotic performance analysis in the context of IVA: Under the ``small-errors" analysis in the asymptotic regime, we provide closed-form expressions for the resulting ISRs. In the particular case where the sources are indeed Gaussian and the covariance matrices are known, these expressions coincide with the induced Cram\'er-Rao Lower Bound (iCRLB, \cite{doron2007cramer}) on the ISRs, given implicitly (as the Fisher information matrix elements) in \cite{weiss2017the}. However (and more importantly), our results also predict the performance in a ``quasi-ML" framework, when the actual distribution and/or covariance matrices are different from their presumed values.
	\addra{\item ``Universal" asymptotic ISR of the Gaussian QML\addra{E}\delra{ estimate} for IVA: We provide a Theorem (and its constructive proof) which states that the asymptotic ISR attained by the Gaussian QML\addra{E}\delra{ estimate} does \textit{not} depend on the sources' full distributions, but rather on their Second-Order Statistics (SOS) \textit{only} (and on some parameters related to the presumed hypothesized model, which will be specifically stated in the sequel).
		\item A lower bound on the Gaussian QMLE ISR - We show that the Gaussian QMLE enjoys the appealing ISR-equivariance property, which in turn, when combined with the Theorem mentioned above, leads to a lower bound on the ISR attained by the Gaussian QMLE - the Gaussian iCRLB - for mixtures with \textit{any} sources' distributions (under some mild conditions, stated explicitly in the paper).}
\end{itemize}

The rest of this paper is organized as follows.\addra{ The ending part of this section is devoted to notations. }In Section \ref{sec_ProblemFormulation} we present the semi-blind Gaussian IVA problem formulation along with the resulting extended SeDJoCo (likelihood) equations. Some important properties of the solution are outlined, focusing on the consistency of (the solution as) the QML\addra{E}\delra{ estimate}s and the implied necessary conditions. Section \ref{sec_AnalyticalAnalysis} is dedicated to the analytical performance analysis of the solution, deriving (approximate) estimation-error terms, followed by the resulting (approximate) ISR. Simulations results are presented in Section \ref{sec_Simulation}, corroborating our analytical results, and Section \ref{sec_Conclusion} is dedicated to conclusions and to final remarks.
\vspace{-0.5cm}
\addra{\subsection{Notations and Preliminaries}
	We use $a, \ua$ and $\A$ for a scalar, column vector and matrix, respectively, where $A_{ij}$ denotes the $(i,j)$-th element of the matrix $\A$ and $a[i]$ denotes the $i$-th element of the vector $\ua$. The superscripts $(\cdot)^{\tps}$ and $(\cdot)^{-1}$ denote the transposition and inverse operators, respectively. The notations $E[\cdot], \Tr(\cdot), {\rm cum}(\cdot,\cdot,\cdot,\cdot)$ and $\|\cdot\|_2$ denote the expectation, trace, $4$-th order joint cumulant \cite{brillinger2001time} and $\ell^2$-norm of their arguments, respectively. The convolution operator is denoted by $\ast$. We also denote by $\I_{K}$ the $K\times K$ identity matrix, and the pinning vector $\ue_k$ denotes the $k$-th column of $\I_{K}$. Using these notations, we define $\E_{ij}\triangleq\ue_i\ue_j^{\tps}$ and $\delta_{ij}\triangleq\ue_i^{\tps}\ue_j$. We also define $\text{vec}(\cdot)$ as the operator which concatenates the columns of an $M \times N$ matrix into an $MN \times 1$ column vector. Furthermore, we define the operator $\text{Bdiag}(\cdot,\ldots,\cdot)$, which creates an $ML\times ML$ block-diagonal matrix from its $L$, $M\times M$ matrix arguments. Finally, the all zeros-matrix (with proper dimensions) is denoted by $\O$.
}

\section{Problem Formulation and the Solution}
\label{sec_ProblemFormulation}
\subsection{Gaussian IVA and the Extended SeDJoCo Transformation}
Consider $M$ datasets of linear, static, memoryless mixtures
\begin{equation} \label{IVA_model}
	\X^{(m)}=\A^{(m)}\mS^{(m)},\;\;\; \forall m\in\{1,\ldots,M\},
\end{equation}
where $\mS^{(m)}=\left[\us_1^{(m)}\;\cdots\;\us_K^{(m)}\right]^{\rm{T}}\in\Rset^{K\times T}$ denotes a matrix of $K$ source signals of length $T$ (for all $m\in\{1,\ldots,M\}$), belonging to
the $m$-th out of $M$ datasets. In each dataset the sources are mixed by an unknown (deterministic) respective mixing-matrix $\A^{(m)}\in\Rset^{K\times K}$, and the
observed mixture signals are given by $\X^{(m)}\in\Rset^{K\times T}$. Based on the observed mixtures datasets $\left\{\X^{(m)}\right\}_{m=1}^{M}$, it is desired to estimate all
$M$ mixing-matrices and thereby recover the source signals. In the same manner as in the standard ICA model, in IVA, too, the sources within each dataset are assumed to be
mutually statistically independent. Clearly, IVA amounts to $M$ independent standard ICA problems when no statistical dependence between source signals across different datasets
exists. However, in IVA statistical dependence between respective sources from different datasets is considered, i.e., the vector $\us_k^{(m_1)}$ may depend on the vector
$\us_k^{(m_2)}$ (for all $m_1,m_2\in\{1,\ldots,M\}$ and all $k\in\{1,\ldots,K\}$), but any two vectors $\us_{k_1}^{(m_1)}$ and $\us_{k_2}^{(m_2)}$ are statistically independent
when $k_1 \neq k_2$ for any $m_1,m_2 \in\{1,\ldots,M\}$.

It turns out \cite{weiss2017the} that in the semi-blind Gaussian model, where the sources are zero-mean Gaussian with known and distinct temporal covariance matrices, the resulting likelihood equations for obtaining the ML\addra{E}\delra{ estimate}s of the matrices $\left\{\B^{(m)}\triangleq {\left[\A^{(m)}\right]^{-1}}\right\}_{m=1}^{M}$ in the IVA problem require a solution of the so-called ``extended SeDJoCo" problem as follows. Let us denote by $\C_k^{(m_1,m_2)} \triangleq E\left[\us_k^{(m_1)} {\us_{k}^{(m_2)\rm{T}}}\right]\in\Rset^{T\times T}$ the temporal covariance matrices between the $k$-th source of the $m_1$-th dataset and the $k$-th source of the $m_2$-th dataset. Now define the $k$-th Source Component Vector (SCV) as $\bs_k \triangleq \text{vec}(\mS_k)$, where $\mS_k$ is the $k$-th source component matrix, defined as $\mS_k \triangleq \left[\us_k^{(1)}\;\cdots\;\us_k^{(M)}\right]^{\rm{T}} \in \Rset^{M \times T}\addra{, \forall k\in\{1,\ldots,K\}}$\delra{, for every $k\in\{1,\ldots,K\}$}\delra{, and where the $\text{vec}(\cdot)$ operator concatenates the columns of an $M \times T$ matrix into an $MT \times 1$ vector}. The covariance matrix of each SCV is given by
\begin{equation} \label{SCM_def}
	\bC_k \triangleq E\left[\bs_k\bs_k^{\rm{T}}\right] =\begin{bmatrix}
		\C_k^{(1,1)} & \cdots & \C_k^{(1,M)}\\
		\vdots & \ddots & \vdots\\
		\C_k^{(M,1)} & \cdots & \C_k^{(M,M)}\end{bmatrix} \in \Rset^{MT \times MT},
\end{equation}
and we denote the respective block-partition of its inverse as
\begin{equation}\label{invereseblockmatrices}
	\bC_k^{-1} \triangleq \begin{bmatrix}
		\P_k^{(1,1)} & \cdots & \P_k^{(1,M)}\\
		\vdots & \ddots & \vdots\\
		\P_k^{(M,1)} & \cdots & \P_k^{(M,M)}\end{bmatrix} \triangleq \bP_k,
\end{equation}
where $\P_k^{(m_1,m_2)} \in \Rset^{T \times T}$, to be used below. Using the fact that the mixtures, being linear combinations of (jointly) Gaussian random vectors, are also (jointly) Gaussian, we can explicitly write the (log-)likelihood function of the given data $\left\{\X^{(m)}\right\}_{m=1}^{M}$ with respect to (w.r.t.) the demixing matrices $\left\{\B^{(m)}\right\}_{m=1}^{M}$. After differentiating w.r.t. $\B^{(m)}$ for every $m\in\{1,\ldots,M\}$ and equating to $\O \in \Rset^{K \times K}$\delra{, the all-zeros matrix, }(see \cite{weiss2017the} for the full detailed derivation)\addra{,} we obtain the likelihood equations,
\begin{multline}\label{e_sedjoco}
	\sum_{\ell=1}^{M}{\hB_{\text{ML}}^{(m)}\hQ_k^{(m,\ell)}{\hB_{\text{ML}}}^{(\ell)\rm{T}}}\ue_k=\ue_k,\\
	\forall k \in \{1,\ldots,K\}, \forall m \in \{1,\ldots,M\},
\end{multline}
where $\hB_{\text{ML}}^{(m)}$ is the ML\addra{E}\delra{ estimate} of $\B^{(m)}$\delra{,}\addra{ and} the matrices $\hQ_{k}^{(m_1,m_2)}$, termed the ``target-matrices", are defined as
\begin{multline}\label{target_matrices_def}
	\hQ_{k}^{(m_1,m_2)} \triangleq \frac{1}{T}\X^{(m_1)}\P_{k}^{(m_1,m_2)}{\X^{(m_2)\rm{T}}}\in \Rset^{K\times K},\\
	\forall k \in \{1,\ldots,K\}, \forall m_1,m_2 \in \{1,\ldots,M\}\addra{.}\delra{,}
\end{multline}
\delra{and where the pinning vector $\ue_k$ denotes the $k$-th column of the $K\times K$ identity matrix. }Thus, the likelihood equations \eqref{e_sedjoco} take the form of the extended SeDJoCo equations, which are formulated more generally as follows:

\noindent \textit{Given $KM^2$ target-matrices $\left\{\Q_k^{(m_1,m_2)}\right\}$, $k \in \{1,\ldots,K\}$, $m_1,m_2 \in \{1,\ldots,M\}$, find a set of $M$ $K \times K$ matrices $\left\{\B^{(m)}\right\}_{m=1}^M$, such that}
\begin{multline}\label{extended_sedjoco_original}
	\left[\sum_{\ell=1}^{M}{\B^{(m)}\Q_k^{(m,\ell)}{\B^{(\ell)\rm{T}}}}\right]\ue_k \triangleq \D_{k}^{(m)}\ue_k =\ue_k,\\
	\forall k \in \{1,\ldots,K\}, \forall m \in \{1,\ldots,M\}.
\end{multline}
The meaning of this statement is that the $k$-th column of each transformed matrix $\D_{k}^{(m)}$, has a ``drilled"\footnote{We say that a $K\times K$ matrix  $\D$ is ``drilled" along its $k$-th column (or row), if that column (or row) is all-zeros, except for its $k$-th element ($D_{kk}$).} structure, hence the name of this transformation.
\vspace{-0.3cm}
\subsection{Properties of the Extended SeDJoCo Solution}
\label{properties_of_e_sedjoco}
Two general (context-free) important properties of solutions of the extended SeDJoCo problem \eqref{extended_sedjoco_original} are \textit{existence} and \textit{non-uniqueness}. As we have shown in \cite{weiss2017the}, when the matrices
\begin{equation} \label{augmented_taget}
	\Omeg_k \triangleq \begin{bmatrix}
		\Q_k^{(1,1)} & \cdots & \Q_k^{(1,M)}\\
		\vdots & \ddots & \vdots\\
		\Q_k^{(M,1)} & \cdots & \Q_k^{(M,M)}\end{bmatrix} \in \Rset^{KM \times KM},
\end{equation}
are Positive Definite (PD) for all $k\in\{1,\ldots,K\}$, a solution of \eqref{extended_sedjoco_original} is guaranteed to exist. By recalling the definition of the target-matrices in \eqref{target_matrices_def}, it is easily (and not surprisingly) seen that a solution always exists in the context of our IVA problem\footnote{In fact, the proof in \cite{weiss2017the} even shows that the ML\addra{E}\delra{ estimate} always exists} (under the general \delra{Second-Order Statistics (}SOS\delra{) }-based IVA identifiability conditions). However, it has been shown \cite{yeredor2016multiple}, \cite{weiss2017the} that the solution is generally not unique, and might therefore lead to a local (rather than the global) maximizer of the likelihood function. Nevertheless, an identification-correction scheme of non-optimal solutions (which are not the ML\addra{E}\delra{ estimate}) was derived in \cite{weiss2017amaximum} for the case of $M=1$ dataset (i.e., for ICA, but can be readily extended to IVA with $M>1$). The scheme first identifies whether a given solution is the global maximizer of the likelihood or not, and, if needed, applies a correction which leads it (with high probability) to the correct (global maximizer) ML\addra{E}\delra{ estimate}.

Next, we address the issue of consistency, a property of an extended SeDJoCo solution in the context of the QML approach for IVA. We begin by noting that for $\left\{\A^{(m)}=\I_K\right\}_{m=1}^M$ (a ``non-mixing" condition\delra{, where $\I_K$ denotes the $K\times K$ identity matrix}), the QML target-matrices $\hQ_{k}^{(m_1,m_2)}$ \eqref{target_matrices_def}, with a set of matrices $\left\{\tP_k^{(m_1,m_2)}\right\}$ replacing the true $\left\{\P_k^{(m_1,m_2)}\right\}$ as defined in \eqref{invereseblockmatrices}, are {\em asymptotically diagonal} under some mild conditions stated in the Lemma below. This important property will be used later to establish consistency of the resulting estimates under mixing conditions.
\begin{lem}[Asymptotic diagonality of the target-matrices when $\left\{\X^{(m)}=\mS^{(m)}\right\}_{m=1}^M$]
	Let us temporarily denote the observation-length-dependent covariance matrices and the presumed respective blocks as in \eqref{SCM_def} and \eqref{invereseblockmatrices} as $\C_{k[T]}^{(m_1,m_2)},\tP_{k[T]}^{(m_1,m_2)}\in\Rset^{T\times T}$, respectively, for an observation length $T$.
	Consider the following three conditions:
	\begin{enumerate}
		\item \label{cond1}
		The following limits exist and are finite:
		\begin{multline*}
			\phi_{k\delra{,}\ell}^{(m_1,m_2)}\triangleq\lim_{T\to\infty}\frac{1}{T}{\rm Tr}\left(\tP_{k[T]}^{(m_1,m_2)}\C_{\ell[T]}^{(m_2,m_1)}\right),\\
			\forall k,\ell\in\{1,\ldots,K\}.
		\end{multline*}
		\item \label{cond2}
		All matrices $\C_{k[T]}^{(m_1,m_2)},\tP_{k[T]}^{(m_1,m_2)}$ can be element-wise bounded by an exponentially-decaying Toeplitz matrix, namely there exist some finite $\rho$ and a positive $\alpha$, such that
		\begin{multline*}
			\left|C_{k[T]\addra{,\tau_1\tau_2}}^{(m_1,m_2)}\delra{[\tau_1,\tau_2]}\right|, \left|\widetilde{P}_{k[T]\addra{,\tau_1\tau_2}}^{(m_1,m_2)}\delra{[\tau_1,\tau_2]}\right|< \rho^2\cdot e^{-\alpha|\tau_1-\tau_2|},\\
			\forall k\in\{1,\ldots,K\}, \forall m_1,m_2\in\{1,\ldots,M\},\\
			\forall \tau_1,\tau_2\in\{1,\ldots,T\}, \forall T\in\Zset^+\delra{,}\addra{.}
		\end{multline*}
		\delra{where $C_{k[T]}^{(m_1,m_2)}[\tau_1,\tau_2]$ and $\widetilde{P}_{k[T]}^{(m_1,m_2)}[\tau_1,\tau_2]$ denote the $(\tau_1,\tau_2)$-th element of $\C_{k[T]}^{(m_1,m_2)}$ and $\tP_{k[T]}^{(m_1,m_2)}$, resp..}
		\item \label{cond3}
		The $4$-th order joint cumulants of each SCV can be similarly bounded by an exponentially decaying function of the time differences, i.e., there exist some finite $\varrho$ and a positive $\beta$, such that for each $k\in\{1,\ldots,K\}$
		\begin{multline*}
			\left|{\rm cum}(s^{(m_1)}_k[\tau_1],s^{(m_2)}_k[\tau_2],s^{(m_3)}_k[\tau_3],s^{(m_4)}_k[\tau_4])\right|\\
			<\varrho^4\cdot e^{-\beta(|\tau_1-\tau_2|+|\tau_3-\tau_4|+|\tau_1-\tau_3|+|\tau_2-\tau_4|+|\tau_1-\tau_4|+|\tau_2-\tau_3|)}
		\end{multline*}
		for all $k,m_1,m_2,m_3,m_4$ and for all $\tau_1,\tau_2,\tau_3,\tau_4$\delra{, where ${\rm cum}(\cdot,\cdot,\cdot,\cdot)$ denotes the $4$-th order joint cumulant \cite{brillinger2001time} of its arguments}.
	\end{enumerate}
	Under these conditions the following property holds:\\
	For $\left\{\A^{(m)}=\I_K\right\}_{m=1}^M$ (so that $\left\{\X^{(m)}=\mS^{(m)}\right\}_{m=1}^M$), the QML target-matrices $\hQ_{k}^{(m_1,m_2)}$ are asymptotically diagonal for all $k\in\{1,\ldots,K\}$ and all $m_1,m_2\in\{1,\ldots,M\}$,
	\begin{multline}
		\left\{\A^{(m)}=\I_K\right\}_{m=1}^M: \\
		\hQ_{k}^{(m_1,m_2)}=\frac{1}{T}\mS^{(m_1)}\tP_{k[T]}^{(m_1,m_2)}\mS^{(m_2)\rm{T}}\xrightarrow[T\rightarrow\infty]{m.s.} \mPhi_k^{(m_1,m_2)},
	\end{multline}
	where $\mPhi_k^{(m_1,m_2)}$ is a diagonal matrix with $\phi_{k\delra{,}1}^{\addra{(m_1,m_2)}}, \phi_{k\delra{,}2}^{\addra{(m_1,m_2)}}, \ldots, \phi_{k\delra{,}K}^{\addra{(m_1,m_2)}}$ along its diagonal, and where the convergence is in the mean square sense \cite{karr1993convergence}.
\end{lem}

Note that the conditions of the Lemma are quite loose and are readily satisfied as long as the temporal covariances of the SCVs and their presumed inverses, $\bC_k$ and $\tbP_k$, respectively, as well as the joint cumulants tensors, all have bounded diagonals and a sufficient rate of decay of their elements away from their diagonals.

The proof of this Lemma is rather straightforward and technical, and is therefore omitted from here due to length considerations. However, note only that the mean of\addra{ the element} $\delra{\hQ}\addra{\widehat{Q}}_{k\addra{,p\ell}}^{(m_1,m_2)}\delra{[p,\ell]}=\frac{1}{T}\us_p^{(m_1)\rm T}\tP_{k[T]}^{(m_1,m_2)}\us^{(m_2)}_{\ell}$ reads
\begin{multline}
	E\left[\delra{\hQ}\addra{\widehat{Q}}_{k\addra{,p\ell}}^{(m_1,m_2)}\delra{[p,\ell]}\right]=\frac{1}{T}{\rm Tr}\left(\tP_{k[T]}^{(m_1,m_2)}E\left[\us^{(m_2)}_{\ell}\us_p^{(m_1)\rm T}\right]\right)=\\
	\begin{cases}
		\frac{1}{T}{\rm Tr}\left(\tP_{k[T]}^{(m_1,m_2)}\C_{\ell[T]}^{(m_2,m_1)}\delra{]}\right)\xrightarrow[T\rightarrow\infty]{}\phi_{k\delra{,}\ell}^{(m_1,m_2)}\addra{,} & p=\ell\\
		0\addra{,} & p\ne \ell
	\end{cases},
\end{multline}
and using Conditions \ref{cond2} and \ref{cond3} it can be shown that their variances tend to zero (note further that the exponential decay of the bound in these conditions is sufficient but not necessary). This important property of asymptotic diagonality establishes the consistency of estimates obtained as extended SeDJoCo solutions. To see this, we note the following.

Let $\left\{\hB^{(m)}_{\text{\normalfont QML}}\right\}_{m=1}^M$ denote a (specific) solution of the extended SeDJoCo equations \eqref{e_sedjoco} with target-matrices defined as in \eqref{target_matrices_def}. By Lemma 1 we have, for all $k\in\left\{1,\ldots,K\right\}$ and for all $m_1,m_2\in\left\{1,\ldots,M\right\}$,
\begin{align}\label{QMLglobalmatrices}
	&\hB^{(m_1)}_{\text{\normalfont QML}}\hQ^{(m_1,m_2)}_k\hB^{(m_2)\tps}_{\text{\normalfont QML}}=\\
	&\hB^{(m_1)}_{\text{\normalfont QML}}\left(\frac{1}{T}\mX^{(m_1)}\P^{(m_1,m_2)}_k\X^{(m_2)\tps}\right)\hB_{\text{\normalfont QML}}^{(m_2)\tps}=\\
	&\hB^{(m_1)}_{\text{\normalfont QML}}\left(\frac{1}{T}\A^{(m_1)}\mS^{(m_1)}\P^{(m_1,m_2)}_k\mS^{(m_2)\tps}\A^{(m_2)\tps}\right)\hB_{\text{\normalfont QML}}^{(m_2)\tps}\triangleq\\
	&\hG^{(m_1)}\left(\frac{1}{T}\mS^{(m_1)}\P^{(m_1,m_2)}_k\mS^{(m_2)\tps}\right)\hG^{(m_2)\tps}\xrightarrow[T\rightarrow\infty]{\rm Lemma\; 1}\\\label{QMLglobalmatriceslast}
	&\hG^{(m_1)}\mPhi^{(m_1,m_2)}_k\hG^{(m_2)\tps},
\end{align}
where we have defined $\hG^{(m)}\triangleq\hB^{(m)}_{\text{\normalfont QML}}\A^{(m)}$ as the QML estimated global demixing-mixing matrix of the $m$-th dataset. Note that by \eqref{QMLglobalmatrices}-\eqref{QMLglobalmatriceslast} we have actually shown that the QML\addra{E}\delra{ estimate} of the demixing matrices shares the ISR-equivariance property (\cite{yeredor2010blind,weiss2017the}), meaning that its resulting separation performance (in terms of ISR) is independent of the mixing matrices (or, put differently, is the same for \textit{any} set of mixing matrices). Thus, returning to \eqref{QMLglobalmatrices}-\eqref{QMLglobalmatriceslast}, since $\hB^{(m)}_{\text{\normalfont QML}}$ is a solution of \eqref{e_sedjoco}, the implied asymptotic extended SeDJoCo equations, expressed in terms of $\left\{\hG^{(m)}\right\}_{m=1}^M$, take the form
\begin{multline}\label{aympextendedSeDJoCo}
	\sum_{\ell=1}^{M}\hG^{(m)}\mPhi^{(m,\ell)}_k\hG^{(\ell)\tps}\ue_k=\ue_k,\\
	\forall k \in \{1,\ldots,K\}, \forall m \in \{1,\ldots,M\}.
\end{multline}
It is trivial to show that for $M=1$ (ICA) and $M=2$ a set of diagonal matrices (which implies perfect separation) solves \eqref{aympextendedSeDJoCo} (for the case $M=1$, see the explicit solution in \cite{weiss2018onconsistency}), meaning that the estimates $\left\{\hB^{(m)}_{\text{\normalfont QML}}\right\}_{m=1}^M$ are indeed consistent. For $M>3$,\delra{ although we currently do not have a proof that a set of diagonal matrices solves \eqref{aympextendedSeDJoCo}, we have witnessed empirically in a variety of simulation scenarios (some of which presented in Section \ref{sec_Simulation}) that this is usually the case. Indeed,} it is easily seen that any set of diagonal matrices $\left\{\hG^{(m)}_o\right\}_{m=1}^M$ \addra{potentially }solves \eqref{aympextendedSeDJoCo} for all $k\neq j\in\{1,\ldots,K\}$ and for all $m\in\{1,\ldots,M\}$. The only \delra{$KM$}\addra{$MK$} \delra{scalar }equations left to solve (for all the \delra{$KM$}\addra{$MK$} diagonal elements of $\left\{\hG^{(m)}_o\right\}_{m=1}^M$) are
\begin{multline}\label{aympextendedSeDJoCo3}
	\sum_{\ell=1}^{M}\widehat{G}_{o\addra{,kk}}^{(m)}\delra{[k,k]}\phi_{k\delra{,}k}^{(m,\ell)}\widehat{G}_{o\addra{,kk}}^{(\ell)}\delra{[k,k]}=1,\\
	\forall k \in \{1,\ldots,K\}, \forall m \in \{1,\ldots,M\},
\end{multline}
which are anyhow related to the inherent ambiguity (in a fully blind scenario) of the scale of the unobserved sources. \addra{Thus, it immediately follows that if a solution to \eqref{aympextendedSeDJoCo3} exists, a perfect separating solution exists to the asymptotic quasi-likelihood equations. Although we currently do not have a proof that a solution to \eqref{aympextendedSeDJoCo3} (which is a set of $MK$ nonlinear, quadratic equations in $MK$ unknowns) always exists, we have witnessed empirically in a variety of simulation scenarios (some of which are presented in Section \ref{sec_Simulation}) that this is usually the case. }In conclusion, under the mild conditions stated in Lemma 1, and assuming the fundamental identifiability conditions hold, a solution of the extended SeDJoCo equations serves as a consistent estimate of \delra{the}\addra{a} set of\addra{ (perfectly)} separating matrices in a general (temporally-diversed) IVA problem.

Having addressed the issue of consistency, thereby justifying the assumption that asymptotically the solution resides in a ``close" vicinity to \addra{an }exact separati\delra{o}n\addra{g solution}, we now turn to assess the asymptotic performance analysis of this solution. In the next section we address the analytical aspects of the performance analysis starting by a ``first-order" perturbation analysis of the extended SeDJoCo solution, followed by the resulting (approximated) ISR in the context of \delra{an }IVA\delra{ problem}. 

\section{Analytical Performance Analysis of the Extended SeDJoCo}
\label{sec_AnalyticalAnalysis}

\subsection{Derivatives of the Solution w.r.t. the Target-Matrices}
Let us define the following matrix
\begin{equation}\label{e_sedjoco_function}
	\F \triangleq \left[\F_1 \; \F_2 \; \cdots \; \F_M\right]\in \Rset^{K\times KM},
\end{equation}
where
\begin{multline}\label{mth_e_sedjoco_function}
	\F_m \triangleq \sum_{k=1}^{K}\sum_{\ell=1}^{M}\E_{kk}{\B^{(\ell)}\Q_k^{(\ell,m)}{\B^{(m)\rm{T}}}}-\I_{K}\in \Rset^{K\times K}, \\
	\forall m\in\{1,\ldots,M\}\addra{.}\delra{,}
\end{multline}
\delra{and $\E_{ij}\triangleq\ue_i\ue_j^{\rm{T}}$. }\addra{Clearly, }$\F$ is\delra{ clearly} a function of all the $\left\{\B^{(m)}\right\}_{m=1}^{M}$ and $\left\{\Q_k^{(m_1,m_2)}\right\}$ matrices. When the target-matrices $\left\{\Q_k^{(m_1,m_2)}\right\}$ are fixed, the solutions (in terms of the $\left\{\B^{(m)}\right\}_{m=1}^M$ matrices) of the equation $\F=\O\delra{\in \Rset^{K\times KM}}$ are the extended SeDJoCo solutions, induced by this fixed set of target-matrices \delra{(where $\O\in \Rset^{K\times KM}$ is the all zeros-matrix)}. We are interested in the perturbations of the elements of the solution matrices $\left\{\B^{(m)}\right\}_{m=1}^{M}$ caused by a perturbation in the elements of the target-matrices $\left\{\Q_k^{(m_1,m_2)}\right\}$. To this end, we first concentrate on the perturbations of the elements of $\left\{\B^{(m)}\right\}_{m=1}^{M}$ caused by a perturbation in the single element $Q_{k,ij}^{(m_1,m_2)}$. Since $\F$ must maintain its zero value of $\O\in\Rset^{K\times KM}$ under these perturbations, exploiting its total derivative w.r.t. $Q_{k,ij}^{(m_1,m_2)}$ we obtain the following equation
\begin{equation}\label{total_derivative}
	\frac{d\F}{dQ_{k,ij}^{(m_1,m_2)}} =\frac{\partial\F}{\partial Q_{k,ij}^{(m_1,m_2)}}+\sum_{k=1}^{K}\sum_{p,q=1}^{M}\frac{\partial\F}{\partial B_{pq}^{(m)}}\cdot\frac{dB_{pq}^{(m)}}{dQ_{k,ij}^{(m_1,m_2)}}=\O.
\end{equation}
Let us carefully compute each term of equation \eqref{total_derivative}. First, notice that
\begin{equation}\label{first_term_matrix}
	\frac{\partial\F}{\partial Q_{k,ij}^{(m_1,m_2)}}=\left[\frac{\partial\F_1}{\partial Q_{k,ij}^{(m_1,m_2)}} \; \frac{\partial\F_2}{\partial Q_{k,ij}^{(m_1,m_2)}} \; \cdots \; \frac{\partial\F_M}{\partial Q_{k,ij}^{(m_1,m_2)}}\right],
\end{equation}
so it is enough to compute each block of \eqref{first_term_matrix}\addra{. After differentiation and some straight forward algebraic simplifications (see Appendix \ref{appendix_RQ} for the detailed computation), we obtain}:
{\delralign{\begin{align}\label{first_term_mth_matrix_old}
			\begin{split}
				&\frac{\partial\F_m}{\partial Q_{k,ij}^{(m_1,m_2)}}= \sum_{\tilde{k}=1}^{K}\sum_{\ell=1}^{M}\E_{\tilde{k}\tilde{k}}{\B^{(\ell)}\frac{\partial\Q_{\tilde{k}}^{(\ell,m)}}{\partial Q_{k,ij}^{(m_1,m_2)}}{\B^{(m)\rm{T}}}}=\\
				&\delta_{mm_2}\E_{kk}\tb_i^{(m_1)}\tb_j^{(m_2)\rm{T}}+\\
				&\left\{\begin{array}{ll}
					\delta_{mm_1}\E_{kk}\tb_j^{(m_2)}\tb_i^{(m_1)\rm{T}} & \mbox{$m_1\neq m_2$}\\
					\delta_{mm_2}\E_{kk}\tb_j^{(m_2)}\tb_i^{(m_1)\rm{T}} & \mbox{$m_1=m_2,i\neq j$}\\
					0 & \mbox{$m_1=m_2,i=j$} \end{array} \right.,
			\end{split}
	\end{align}}
	{\addra{\begin{align}\label{first_term_mth_matrix}
				\frac{\partial\F_m}{\partial Q_{k,ij}^{(m_1,m_2)}}=&\delta_{mm_2}\E_{kk}\tb_i^{(m_1)}\tb_j^{(m_2)\rm{T}}+\nonumber\\
				&\left\{\begin{array}{ll}
					\delta_{mm_1}\E_{kk}\tb_j^{(m_2)}\tb_i^{(m_1)\rm{T}}, & \mbox{$m_1\neq m_2$}\\
					\delta_{mm_2}\E_{kk}\tb_j^{(m_2)}\tb_i^{(m_1)\rm{T}}, & \mbox{$m_1=m_2,i\neq j$}\\
					0, & \mbox{$m_1=m_2,i=j$} \end{array} \right.,\\
				\frac{\partial\F_n}{\partial B_{pq}^{(m)}}=&\E_{pq}\Q_p^{(m,n)}\B^{(n)\rm{T}}+\nonumber\\
				&\delta_{mn}\sum_{k=1}^{K}\sum_{\ell=1}^{M}\E_{kk}\B^{(\ell)}\Q_k^{(\ell,m)}\E_{qp},\label{second_term_mth_matrix}
	\end{align}}}where $\tb_i^{(m)}$ denotes the $i$-th column of $\B^{(m)}$\delra{ and $\delta_{mn}$ equals $1$ when $m=n$, and equals zero otherwise}.\delra{ Next, we compute the derivative of the second term in \eqref{total_derivative}. Again, we do this by first computing for the $n$-th block, that is
		\begin{align}\label{second_term_mth_matrix_old}
			\begin{split}
				&\frac{\partial\F_n}{\partial B_{pq}^{(m)}} = \sum_{k=1}^{K}\sum_{\ell=1}^{M}\E_{kk}{\frac{\partial}{\partial B_{pq}^{(m)}}}\left(\B^{(\ell)}\Q_{k}^{(\ell,n)}\B^{(n)\rm{T}}\right)=\\
				&\sum_{k=1}^{K}\E_{kk}\Bigg[\delta_{mn}\left(\B^{(m)}\Q_k^{(m,m)}\E_{qp}+\E_{pq}\Q_k^{(m,m)}+\right.\\
				&\sum_{\substack{\ell=1\\\ell\neq m}}^{M}\left.\B^{(\ell)}\Q_k^{(\ell,m)}\E_{qp}\right)+\left(1-\delta_{mn}\right)\left(\E_{pq}\Q_k^{(m,n)}\B^{(n)\rm{T}}\right)\Bigg]=\\
				&\sum_{k=1}^{K}\E_{kk}\left(\E_{pq}\Q_k^{(m,n)}\B^{(n)\rm{T}}+\delta_{mn}\sum_{\ell=1}^{M}\B^{(\ell)}\Q_k^{(\ell,m)}\E_{qp}\right),
			\end{split}
		\end{align}
		where we have used (\cite{petersen2008matrix})
		\begin{align}
			&\frac{\partial\mX\B\mX^{\rm{T}}}{\partial X_{ij}}=\mX\B\E_{ji}+\E_{ij}\B\mX^{\rm{T}},\\
			&\frac{\partial\mX\A}{\partial X_{ij}}=\E_{ij}\A.
\end{align}}}
Now, for ease of exposition, define
\begin{equation}\label{definition_of_Y_and_H}
	\Y\triangleq\frac{\partial\F}{\partial Q_{k,ij}^{(m_1,m_2)}}, \;\;\; \H_{[m,p,q]}\triangleq \frac{\partial\F}{\partial B_{pq}^{(m)}},
\end{equation}
so that equation \eqref{total_derivative} can now be written as
\begin{equation}\label{total_derivative_simple_form}
	-\Y=\sum_{m=1}^{M}\sum_{p,q=1}^{K}\H_{[m,p,q]}\cdot\frac{dB_{pq}^{(m)}}{dQ_{k,ij}^{(m_1,m_2)}}.
\end{equation}
Applying the $\text{vec}(\cdot)$ operator on both sides of \eqref{total_derivative_simple_form} yields
\begin{equation}\label{total_derivative_vec_form}
	-\uy=\sum_{m=1}^{M}\sum_{p,q=1}^{K}\uh_{[m,p,q]}\cdot\frac{dB_{pq}^{(m)}}{dQ_{k,ij}^{(m_1,m_2)}}.
\end{equation}
where $\uy\triangleq\text{vec}\left(\Y\right)$ and $\uh_{[m,p,q]}\triangleq\text{vec}\left(\H_{[m,p,q]}\right)$.
Finally, define
\begin{align}\label{definitionfothetaandH}
	&\H\triangleq\left[\uh_{[1,1,1]} \uh_{[1,2,1]} \cdots \uh_{[1,K,K]} \uh_{[2,1,1]} \uh_{[2,2,1]} \cdots \uh_{[M,K,K]}\right],\\
	&\utheta\triangleq\left[\theta_{[1,1,1]} \theta_{[1,2,1]} \cdots \theta_{[1,K,K]} \theta_{[2,1,1]} \theta_{[2,2,1]} \cdots \theta_{[M,K,K]}\right]^{\rm{T}},
\end{align}
where $\theta_{[m,p,q]}\triangleq\frac{dB_{pq}^{(m)}}{dQ_{k,ij}^{(m_1,m_2)}}$, to obtain equation \eqref{total_derivative} in its compact form
\begin{equation}\label{total_derivative_simplest_form}
	-\uy=\H\utheta.
\end{equation}
It can be shown that $\H\in\Rset^{K^2M\times K^2M}$ is full-rank\footnote{with probability 1} (the proof is straightforward, though rather technical, so we omit this part from the paper). Therefore, the derivatives of the elements of $\left\{\B^{(m)}\right\}_{m=1}^{M}$ w.r.t. the single element $Q_{k,ij}^{(m_1,m_2)}$ are given by
\begin{equation}\label{derivative_solution}
	\utheta=-\H^{-1}\uy.
\end{equation}
We emphasize that the linear system of equations \eqref{derivative_solution} may be solved separately for obtaining the derivative w.r.t. each element $Q_{k,ij}^{(m_1,m_2)}$, i.e., for every combination of $k,i,j\in\{1,\ldots,K\}$ and $m_1,m_2\in\{1,\ldots,M\}$, so as to obtain all the derivatives of all elements of the solution matrices $\left\{\B^{(m)}\right\}_{m=1}^{M}$ w.r.t. all elements of all target-matrices.

\subsection{A ``First-Order" Perturbation Analysis}
Let us now return to our semi-blind Gaussian IVA problem. Assuming the conditions stated in \delra{Theorem}\addra{Lemma} 1 hold, when the target-matrices are
\begin{align}\label{target_matrices_def2}
	\begin{split}
		\Q_{k}^{(m_1,m_2)}=&\frac{1}{T}E\left[\X^{(m_1)}\P_{k}^{(m_1,m_2)}{\X^{(m_2)\rm{T}}}\right]=\lim_{T\to \infty} \hQ_k^{(m_1,m_2)},\\
		&\forall k \in \{1,\ldots,K\}, \forall m_1,m_2 \in \{1,\ldots,M\},
	\end{split}
\end{align}
by the consistency of the ML\addra{E}\delra{ estimate}, a solution to \eqref{e_sedjoco} is $\left\{\hB_{\text{ML}}^{(m)}=\B^{(m)}\right\}_{m=1}^{M}$, i.e., the set of true demixing matrices.
For ease of the exposition, we define
\begin{align}\label{wideQmatrix}
	\tQ &\triangleq \left[\Q_1^{(1,1)} \;\Q_2^{(1,1)}\; \cdots \;\Q_K^{(1,1)} \;\Q_1^{(2,1)} \;\Q_2^{(2,1)}\; \cdots \;\Q_K^{(M,M)}\right],\\
	\uq &\triangleq \text{vec}^*\left(\tQ\right)\in\Rset^{M_q\times 1},
\end{align}
where $\text{vec}^*(\cdot)$ concatenates the columns of the matrix $\tQ$ into a column vector but takes each element duplicated by symmetry only once, on its first occurrence (since $\left\{\Q_k^{(m,m)}\right\}$ are symmetric and $\Q_k^{(m_1,m_2)}=\Q_k^{(m_2,m_1)\rm{T}}$), $M_q=\frac{MK^2(1+MK)}{2}$, and $\htQ,\huq$ are defined in the same manner as $\tQ,\uq$ only with the (finite sample size) target-matrices $\left\{\hQ_k^{(m_1,m_2)}\right\}$, as defined in \eqref{target_matrices_def}. Now, since the elements of the solution matrices of the extended SeDJoCo equations are determined by the elements of the target-matrices, we may write
\begin{equation}\label{e_sedjoco_analytic_function}
	\widehat{B}_{ij}^{(m)}=\tilde{f}_{ij}^{(m)}\left(\huq\right), \forall m\in\{1,\ldots,M\}, \forall i,j\in\{1,\ldots,K\},
\end{equation}
where $\tilde{f}_{ij}^{(m)}:\Rset^{M_q\times 1}\rightarrow\Rset$ denotes the function that maps a given set of target-matrices to the $(i,j)$-th element of the \delra{quasi-}\addra{Q}ML\addra{E}\delra{ estimate} of the matrix $\B^{(m)}$.
Now, define $\uepsilon_{q}\triangleq\huq-\uq$ as the estimation error (vector) in the estimation of $\uq$.
Assuming $\tilde{f}_{ij}^{(m)}$ is an analytical function in the neighborhood of $\huq=\uq$, and that $\|\uepsilon_{q}\|_2$ is ``small enough" (in the sense that the second- and higher-order terms of the Taylor expansion of $\tilde{f}_{ij}^{(m)}$ in the neighborhood of $\huq=\uq$ are negligible w.r.t. the first-order term),\delra{ where $\|\cdot\|_2$ denotes the $\ell^2$-norm,} we may write
\begin{equation}\label{firstordertaylor}
	\widehat{B}_{ij}^{(m)} \approx B_{ij}^{(m)}+\left(\left.\frac{d\widehat{B}_{ij}^{(m)}}{d\huq}\right|_{\footnotesize{\huq=\uq}}\right)\left(\huq-\uq\right)\triangleq B_{ij}^{(m)}+\ug_{ij}^{(m)\rm{T}}\uepsilon_q,
\end{equation}
where we have neglected second- and higher-order terms of the estimation error $\uepsilon_q$, and defined $\ug_{ij}^{(m)}$ as the $(i,j,m)$-th gradient vector of $\widehat{B}_{ij}^{(m)}$ w.r.t. the vector $\huq$, evaluated at $\huq=\uq$. Hence, if we define $\varepsilon_{B\addra{,ij}}^{(m)}\delra{[i,j]}\triangleq \widehat{B}_{ij}^{(m)}-B_{ij}^{(m)}$, equation \eqref{firstordertaylor} may be written as
\begin{equation}\label{errorinB}
	\varepsilon_{B\addra{,ij}}^{(m)}\delra{[i,j]}\approx \ug_{ij}^{(m)\rm{T}}\uepsilon_q.
\end{equation}
\subsection{The resulting ISR in the context of IVA}
\delra{Since i}\addra{I}n the context of IVA we have $E\left[\uepsilon_q\right]=\uo$, \addra{so }we also have
\begin{multline}\label{approximatedunbiasedness}
	E\left[\varepsilon_{B\addra{,ij}}^{(m)}\delra{[i,j]}\right] \approx \ug_{ij}^{(m)\rm{T}}E\left[\uepsilon_q\right]=0, \\ \forall m\in\{1,\ldots,M\},\forall i,j\in\{1,\ldots,K\}.
\end{multline}
Exploiting the ISR-equivariance property of the (Q)ML\addra{E}\delra{ estimate} of the demixing matrices, it suffices to consider and evaluate the resulting ISR, defined as
\begin{multline} \label{ISR_defenition}
	\text{ISR}_{ij}^{(m)} \triangleq E\left[\frac{\left|\left(\hB^{(m)}\A^{(m)}\right)_{ij}\right|^2}{\left|\left(\hB^{(m)}\A^{(m)}\right)_{ii}\right|^2}\right]\cdot \frac{E\left[{\us_j^{(m)\rm{T}}}\us_j^{(m)}\right]}{E\left[{\us_i^{(m)\rm{T}}}\us_i^{(m)}\right]},\\
	\; 1\leq i \neq j\leq K, \forall m \in \{1,\ldots,M\},
\end{multline}
for any particular (arbitrarily chosen) set of mixing matrices $\left\{\A^{(m)}\right\}_{m=1}^{M}$. Thus, we consider the convenient choice of the set of identity matrices, i.e., $\left\{\A^{(m)}=\I_K=\B^{(m)}\right\}_{m=1}^{M}$, and for ease of notation we define the block diagonal matrix $\bA \triangleq \text{Bdiag}\left(\A^{(1)},\ldots,\A^{(M)}\right) \in \Rset^{KM \times KM}$\delra{, where the $\text{Bdiag}(\cdot)$ operator creates a block-diagonal matrix from its square matrix arguments}. Hence, we notice that in this case ($\bA=\I_{KM}$),
{\delralign{\begin{align}\label{globalmatrixatI_old}
			\hT^{(m)}&\triangleq \hB^{(m)}\A^{(m)}=\hB^{(m)}\\
			&=\B^{(m)}+\meps_B^{(m)}=\I_K+\meps_B^{(m)}, \; \forall m\in \{1,\ldots,M\},
\end{align}}}
\addra{\begin{equation}\label{globalmatrixatI}
		\hT^{(m)}\triangleq\hB^{(m)}\A^{(m)}=\hB^{(m)}=\B^{(m)}+\meps_B^{(m)}=\I_K+\meps_B^{(m)},
\end{equation}}
\addra{for all $m\in \{1,\ldots,M\}$}, where $\meps_B^{(m)}\in\Rset^{K\times K}$ denotes the estimation error matrix in the estimation of $\B^{(m)}$ (whose $(i,j)$-th element is $\varepsilon_{B\addra{,ij}}^{(m)}\delra{[i,j]}$), so that
{\delralign{\begin{align}\label{globalmatrixerrortermsinB_old}
			&\widehat{T}^{(m)}_{ij}=\varepsilon_B^{(m)}[i,j],\\
			1\leq i &\neq j\leq K, \forall m \in \{1,\ldots,M\}.\nonumber
\end{align}}}
{\addra{\begin{equation}\label{globalmatrixerrortermsinB}
			\widehat{T}^{(m)}_{ij}=\varepsilon_{B\addra{,ij}}^{(m)},\; 1\leq i\neq j\leq K, \forall m \in \{1,\ldots,M\}.
	\end{equation}}
	\delra{If we a}\addra{A}ssum\delra{e}\addra{ing} that $\varepsilon_{B\addra{,ij}}^{(m)}\delra{[i,j]}\ll 1$ for all $i,j,m$ and, for simplicity of the exposition, that all sources have unit power, it follows that
	\begin{multline}\label{approximatedISR}
		\text{ISR}_{ij}^{(m)}=E\left[\frac{\left|\widehat{T}_{ij}^{(m)}\right|^2}{\left|\widehat{T}_{ii}^{(m)}\right|^2}\right]\approx E\left[\left|\widehat{T}_{ij}^{(m)}\right|^2\right]=E\left[\left|\varepsilon_{B\addra{,ij}}^{(m)}\delra{[i,j]}\right|^2\right],\\ 1\leq i \neq j\leq K, \forall m\in\{1,\ldots,M\}.
	\end{multline}
	Substituting \eqref{errorinB} into \eqref{approximatedISR}, we have that
	\begin{equation}\label{ISRintermsofCq}
		\text{ISR}_{ij}^{(m)}\approx E\left[\left|\varepsilon_{B\addra{,ij}}^{(m)}\delra{[i,j]}\right|^2\right] \approx \ug_{ij}^{(m)\rm{T}}\C_{\hat{q}_{(I)}}\ug_{ij}^{(m)},
	\end{equation}
	where $\C_{\hat{q}_{(I)}}$ denotes the covariance matrix of $\huq$ when $\bA=\I_{KM}$, i.e.,
	\begin{equation}\label{covarianveofQ}
		\C_{\hat{q}_{(I)}}\triangleq \left.E\left[\left(\huq-\uq\right)\left(\huq-\uq\right)^{\rm{T}}\right]\right|_{\scriptsize{\bA=\I_{KM}}}=\left.E\left[\uepsilon_{q}\uepsilon_{q}^{\rm{T}}\right]\right|_{\scriptsize{\bA=\I_{KM}}}.
	\end{equation}
	Let us write $\C_{\hat{q}_{(I)}}$$=E$$\left.\left[\huq\huq^{\rm{T}}\right]-\uq\uq^{\rm{T}}\right|_{\scriptsize{\bA=\I_{KM}}}$ and compute each term separately. First, recall that $\uq=\text{vec}\addra{^*}\left(\tQ\right)$, and when $\bA=\I_{KM}$ we get
	\begin{align}\label{valuesofQelements}
		Q_{k,ij}^{(m_1,m_2)}&=\frac{1}{T}E\left[\us_i^{(m_1)\rm{T}}\P_k^{(m_1,m_2)}\us_j^{(m_2)}\right]\\
		&=\delta_{ij}\frac{1}{T}\Tr\left(\C_i^{(m_2,m_1)}\P_k^{(m_1,m_2)}\right),
	\end{align}
	\delra{where $\Tr(\cdot)$ is the trace operator, }hence all the matrices $\left\{\Q_{k}^{(m_1,m_2)}\right\}$ are diagonal in this case.\addra{ In addition, }\delra{D}\addra{d}efine the indexing function\delra{s}
	\begin{align}\label{indexfunctions}
		\addra{\text{ind}_{\ell}}&\addra{\triangleq i_{\ell}+(j_{\ell}-1)K+(k_{\ell}-1)K^2} \nonumber \\
		&\addra{+(m_{\ell}-1)K^3+(n_{\ell}-1)MK^3, \ell=1,2,}
	\end{align}
	\addra{for all $1\leq i_{\ell},j_{\ell},k_{\ell}\leq K, \; 1\leq m_{\ell},n_{\ell}\leq M$, and denote $\tilde{i}:=\text{ind}[1], \tilde{j}:=\text{ind}[2]$,}
	{\delralign{\begin{align}\label{indexfunctions_old}
				\begin{split}
					\tilde{i} &:= \text{ind1}[i_1,j_1,k_1,m_1,n_1] \\
					&\triangleq i_1+(j_1-1)K+(k_1-1)K^2+(m_1-1)K^3\\
					&+(n_1-1)MK^3, \\
					\tilde{j} &:= \text{ind2}[i_2,j_2,k_2,m_2,n_2] \\
					&\triangleq i_2+(j_2-1)K+(k_2-1)K^2+(m_2-1)K^3\\
					&+(n_2-1)MK^3, \\
					&1 \leq i_1,j_1,k_1,i_2,j_2,k_2\leq K, \; 1 \leq m_1,n_1,m_2,n_2\leq M,
				\end{split}
	\end{align}}}
	so that now we may write \addra{in shorthand}
	\begin{align}\label{single_element_of_qqT}
		&\left(\uq\right)_{\tilde{i}} = \delta_{i_1j_1}\frac{1}{T}\Tr\left(\C_{i_1}^{(n_1,m_1)}\P_{k_1}^{(m_1,n_1)}\right),\nonumber \\
		\addra{\Rightarrow}  &\addra{\left(\uq\uq^{\rm{T}}\right)_{\tilde{i}\tilde{j}} = \delta_{i_1j_1}\delta_{i_2j_2}\cdot} \nonumber\\ &\addra{\frac{1}{T^2}\Tr\left(\C_{i_1}^{(n_1,m_1)}\P_{k_1}^{(m_1,n_1)}\right)\Tr\left(\C_{i_2}^{(n_2,m_2)}\P_{k_2}^{(m_2,n_2)}\right).}
	\end{align}
	\delra{and
		\begin{multline}\label{single_element_of_qqT_old}
			\left(\uq\uq^{\rm{T}}\right)_{\tilde{i}\tilde{j}} = \delta_{i_1j_1}\delta_{i_2j_2}\cdot \\ \frac{1}{T^2}\Tr\left(\C_{i_1}^{(n_1,m_1)}\P_{k_1}^{(m_1,n_1)}\right)\Tr\left(\C_{i_2}^{(n_2,m_2)}\P_{k_2}^{(m_2,n_2)}\right).
	\end{multline}}
	To complete the computation of the elements of $\C_{\hat{q}_{(I)}}$, we now address the term $E\left[\huq\huq^{\rm{T}}\right]$. By definition of the elements of $\huq$, and again using $\bA=\I_{KM}$, we have that
	\begin{multline}\label{secondtermofCq}
		\;\;\left(E\left[\huq\huq^{\rm{T}}\right]\right)_{\tilde{i}\tilde{j}}=\\
		\frac{1}{T^2}E\left[\us_{i_1}^{(m_1)\rm{T}}\P_{k_1}^{(m_1,n_1)}\us_{j_1}^{(n_1)}\us_{i_2}^{(m_2)\rm{T}}\P_{k_2}^{(m_2,n_2)}\us_{j_2}^{(n_2)}\right].
	\end{multline}
	\addra{Finally,}\delra{By} using the linearity\delra{ property} of the expectation and trace operators, exploiting the mutual statistical independence between all SCVs, and subtracting \eqref{single_element_of_qqT}, we\delra{ finally} obtain
	{\delralign{\begin{equation}
				\label{CovairanceofQ_expression_old}
				\begin{aligned}
					&\left(\C_{\hat{q}_{(I)}}\right)_{\tilde{i}\tilde{j}}=\\
					&\begin{cases}
						\Tr\left(\C_{i_1}^{(m_2,m_1)}\P_{k_1}^{(m_1,n_1)}\C_{i_1}^{(n_1,n_2)}\P_{k_2}^{(n_2,m_2)}\right)+\\ \Tr\left(\C_{i_1}^{(n_2,m_1)}\P_{k_1}^{(m_1,n_1)}\C_{i_1}^{(n_1,m_2)}\P_{k_2}^{(m_2,n_2)}\right),\\
						\;\quad\quad\quad\quad\quad\quad\quad\quad\quad\quad\quad\quad\quad\quad\quad\quad i_1=j_1=i_2=j_2 \\
						\Tr\left(\C_{i_1}^{(m_2,m_1)}\P_{k_1}^{(m_1,n_1)}\C_{j_1}^{(n_1,n_2)}\P_{k_2}^{(n_2,m_2)}\right),\\
						\quad\quad\quad\quad\quad\quad\quad\quad\quad\quad\quad\quad\quad (i_1,j_1)=(i_2,j_2), i_1\neq j_1 \\
						\Tr\left(\C_{i_1}^{(n_2,m_1)}\P_{k_1}^{(m_1,n_1)}\C_{i_2}^{(n_1,m_2)}\P_{k_2}^{(m_2,n_2)}\right),\\
						\quad\quad\quad\quad\quad\quad\quad\quad\quad\quad\quad\quad\quad (i_1,j_1)=(j_2,i_2), i_1\neq j_1 \\
						0,\\
						\;\quad\quad\quad\quad\quad\quad\quad\quad\quad\quad\quad\quad\quad\quad\quad\quad\quad\quad\quad \text{otherwise}
					\end{cases}
				\end{aligned},
	\end{equation}}}
	{\addra{\begin{align}\label{CovairanceofQ_expression}
				&\left(\C_{\hat{q}_{(I)}}\right)_{\tilde{i}\tilde{j}}=\nonumber\\
				&\delta_{i_1i_2}\delta_{j_1j_2}\Tr\left(\C_{i_1}^{(m_2,m_1)}\P_{k_1}^{(m_1,n_1)}\C_{j_1}^{(n_1,n_2)}\P_{k_2}^{(n_2,m_2)}\right)+\nonumber\\
				&\delta_{i_1j_2}\delta_{j_1i_2}\Tr\left(\C_{i_1}^{(n_2,m_1)}\P_{k_1}^{(m_1,n_1)}\C_{i_2}^{(n_1,m_2)}\P_{k_2}^{(m_2,n_2)}\right),
	\end{align}}}
	where we have used the assumption that the sources are Gaussian (and in particular, Isserlis' \delra{t}\addra{T}heorem \cite{isserlis1918formula}, regarding the computation of higher-order moments of the multivariate Gaussian distribution) \textit{only} in the computation of the terms for which $i_1=j_1=i_2=j_2$ (see Appendix \ref{appendix_a}).
	
	\delra{In }Appendix \ref{appendix_b} \delra{we present}\addra{contains} the solution for the system of equation\addra{s} \delra{\eqref{total_derivative_simplest_form}}\addra{\eqref{derivative_solution}}, for the case where $\bA=\I_{KM}$. Solving \delra{\eqref{total_derivative_simplest_form}}\addra{\eqref{derivative_solution}} for each element of $\uq$ yields $M_q$ vectors, which will be denoted here by $\left\{\utheta_r\right\}_{r=1}^{M_q}$. Now, notice that if we define the matrix whose columns are $\left\{\utheta_r\right\}_{r=1}^{M_q}$, we have that
	\begin{align}\label{GradientMatrix}
		\begin{split}
			\G &\triangleq \left[\utheta_1 \; \utheta_2 \; \cdots \; \utheta_{M_q}\right] \\
			&=\left[\tilde{\ug}_1 \; \tilde{\ug}_2 \; \cdots \; \tilde{\ug}_{MK^2}\right]^{\rm{T}}\in\Rset^{MK^2\times M_q},
		\end{split}
	\end{align}
	where
	\begin{equation}\label{indexfunctiongradient}
		\tilde{\ug}_{i+(j-1)K+(m-1)K^2} = \ug_{ij}^{(m)},
	\end{equation}
	i.e., the rows of $\G$ are exactly the gradient vectors $\ug_{ij}^{(m)}$ defined in \eqref{firstordertaylor}. \delra{With the closed-form expressions for the gradient vectors $\ug_{ij}^{(m)}$ (obtained by the rows of the matrix $\G$) and for the elements of $\C_{\hat{q}_{(I)}}$ (given in \eqref{CovairanceofQ_expression}), we have obtained closed-form expressions for the resulting (approximated) ISR elements \eqref{ISR_defenition} attained by the (ML) solution of the extended SeDJoCo equations.}\addra{To summarize, we have obtained closed-form expressions for the resulting (approximated) ISR elements \eqref{ISR_defenition} attained by the (QML) solution of the extended SeDJoCo equations. These expressions are computed as follows:
		\begin{enumerate}
			\item Compute the covariance matrix $\C_{\hat{q}_{(I)}}$ (whose elements are given explicitly in \eqref{CovairanceofQ_expression});
			\item Compute the gradient matrix $\G$ (defined in \eqref{GradientMatrix}), whose $M_q$ columns are the solutions of \eqref{total_derivative_simplest_form}, given explicitly in Appendix \ref{appendix_a} (each with its corresponding indices);
			\item Obtain the gradient vectors $\ug_{ij}^{(m)}$ (as the rows of $\G$);
			\item Compute the predicted (approximated) ISRs via \eqref{ISRintermsofCq}.
	\end{enumerate}}
	
	\addra{Fortunately, it turns out that the asymptotic ISR obtained by the (QML) solution of the extended SeDJoCo is determined, and therefore may be accurately predicted, only by the SOS of the sources, rather than by their full distributions, as stated in following Theorem.
		\begin{thm}[``Universal" asymptotic ISR of the Gaussian QMLE for IVA] \label{theorem1}
			For model \eqref{IVA_model}, under the conditions stated in Lemma 1, the asymptotic ISR (defined in \eqref{ISR_defenition}) of the Gaussian QMLE, obtained as a solution of the extended SeDJoCo equations, does not depend on the sources' full distributions, but only on their covariance matrices $\left\{\C_k^{(m_1,m_2)}\right\}$, and the matrices $\left\{\P_k^{(m_1,m_2)}\right\}$ used for the construction of the target-matrices \eqref{target_matrices_def}.
		\end{thm}
		\begin{proof}
			See Appendix \ref{appendix_b}.
			\vspace{-0.2cm}
			\begin{align}
				\vspace{-0.5cm}
				\tag*{$\blacksquare$}
			\end{align}
	\end{proof}}
	\vspace{-0.2cm}
	Thus, the \delra{result we have obtained for the}\addra{(approximated) closed-form} ISR \addra{expressions we have obtained} hold\delra{s} not only when the sources are Gaussian, but also for \textit{any} distribution of the sources. In addition, we emphasize that the matrices $\left\{\P_k^{(m_1,m_2)}\right\}$ in \eqref{CovairanceofQ_expression} need not be the actual blocks obtained by the true covariance matrices, defined in \eqref{invereseblockmatrices}, but rather can be \textit{arbitrarily} chosen matrices (\delra{with}\addra{under} the \delra{restrictions and }conditions stated \delra{and assumed }in \delra{Theorem}\addra{Lemma} 1) used to construct the target-matrices for the extended SeDJoCo equations. 
	
	\addra{Furthermore, this Theorem yields yet another important and informative result regarding the performance of the Gaussian QMLE, which is given as follows. By virtue of Theorem \ref{theorem1}, the asymptotic ISR attained by the Gaussian QMLE is determined only by the true covariances of the sources $\left\{\C_k^{(m_1,m_2)}\right\}$, and by the matrices $\left\{\P_k^{(m_1,m_2)}\right\}$ used for the construction of the target-matrices. Hence, the resulting asymptotic ISRs attained by the Gaussian QMLE are affected only by the mismodeling error introduced by the deviation of the presumed $\left\{\P_k^{(m_1,m_2)}\right\}$ matrices from the actual blocks obtained by the true covariances matrices. But in the ``best" scenario, when the sources' SOS are actually known \textit{a-priori}, the predicted ISRs coincide with the Gaussian iCRLB (\cite{weiss2017the}) \textit{regardless} of the true sources' distributions. Therefore, we conclude that the Gaussian iCRLB serves as an asymptotically attainable lower bound on the resulting ISR of the Gaussian QMLE.}
	
	To summarize, the expressions we have obtained for the resulting (approximated) ISRs\delra{ when using}\addra{ attained by} the extended SeDJoCo \addra{(QML) }solution\delra{ as a set of separating matrices} hold for a general IVA problem, not necessarily Gaussian and not necessarily semi-blind (as long as the mild conditions stated in \delra{Theorem}\addra{Lemma} 1 hold).\addra{ Moreover, the Gaussian iCRLB serves as a lower bound (asymptotically attainable) on the resulting ISR of the Gaussian QMLE, regardless of the true sources' distributions.}

\section{Simulation Results}
\label{sec_Simulation}
In order to corroborate our analytical results, we present simulations results of \delra{two}\addra{three} experiments. In these experiments, we compare analytically predicted results, based on our ``small-errors" assumption, with empirical results. Furthermore, the setup of these\delra{s} experiments describes more realistic semi-blind scenarios, which are more suitable for modeling ``real-life" problems, where some \textit{a-priori} information is available (or assumed), but is likely to be inaccurate\addra{ (or false)}.

In \delra{both}\addra{all the} experiments we consider (\addra{according to }model \eqref{IVA_model}) mixtures of $M$ datasets, each with $K$ stationary sources\delra{ and a sample size of $T\delra{=1000}$}. The $k$-th source of the $m$-th dataset, $s_k^{(m)}[t]$, is generated as
\begin{multline} \label{source_definition_simul}
	s_k^{(m)}[t] = \sum_{\ell=1}^{M}{w_k^{(\ell)}[t]\ast h_k^{(m,\ell)}[t] },\\
	\forall k\in\{1,\ldots,K\}, \forall m\in\{1,\ldots,M\},
\end{multline}
where $\left\{w_k^{(m)}[t]\right\}$ are all mutually independent white noise processes, $\left\{h_k^{(m_1,m_2)}[t]\right\}$ are Finite Impulse Response (FIR) filters of length $L$ for which
\begin{equation} \label{FIR_energy}
	\sum_{t=0}^{L-1}\left|h_k^{(m_1,m_2)}[t]\right|^2 = \begin{cases}
		1, & m_1=m_2 \\
		\eta, & m_1\neq m_2 \\
	\end{cases},
\end{equation}
so that $\eta$ is a parameter which controls the ``relative energy" contained in the cross-spectra between corresponding sources from different datasets\delra{, and $\ast$ denotes the convolution operator}. As can be seen from \eqref{source_definition_simul}, $h_k^{(m_1,m_2)}[t]$ is the FIR filter applied to the $m_2$-th white driving-noise in order to generate a component of the $k$-th source in the $m_1$-th dataset. The cross-spectrum between any pair $\left\{s_k^{(m_1)}[t],s_k^{(m_2)}[t]:m_1 \neq m_2\right\}$ is obviously non-zero when $\eta>0$. In each of the \delra{two}\addra{three} experiments, the solution matrices were obtained using Newton's method for an iterative solution of the extended SeDJoCo \cite{weiss2017the}. In order to focus on the small errors and avoid convergence to false local maxima (which are out of scope of our analysis), we initialized the iterative solution with the true demixing matrices $\left\{\B^{(m)}\right\}_{m=1}^M$. Note that this choice guarantees (with probability 1) that the initial solution is not the expected solution $\left\{\B^{(m)}+\meps_B^{(m)}\right\}_{m=1}^M$, which stems from our analysis (since for every finite sample size $T$ the ISRs are strictly positive). We shall present our results in terms of the total normalized ISR, defined as
\begin{equation} \label{JISR_defenition}
	\text{ISR}_{\text{norm}} \triangleq \frac{1}{MK(K-1)}\sum_{m=1}^{M}{\sum_{\substack{i,j=1 \\ i \neq j}}^{K}{\text{ISR}_{ij}^{(m)}}}.
\end{equation}
\delra{All empirical results were obtained by averaging $10000$ independent trials.}
\vspace{-0.7cm}
\subsection{Gaussian Sources with Inaccurate Covariance Matrices}\label{inaccuratecovariances}
\begin{figure}
	\centering
	\includegraphics[width=0.5\textwidth]{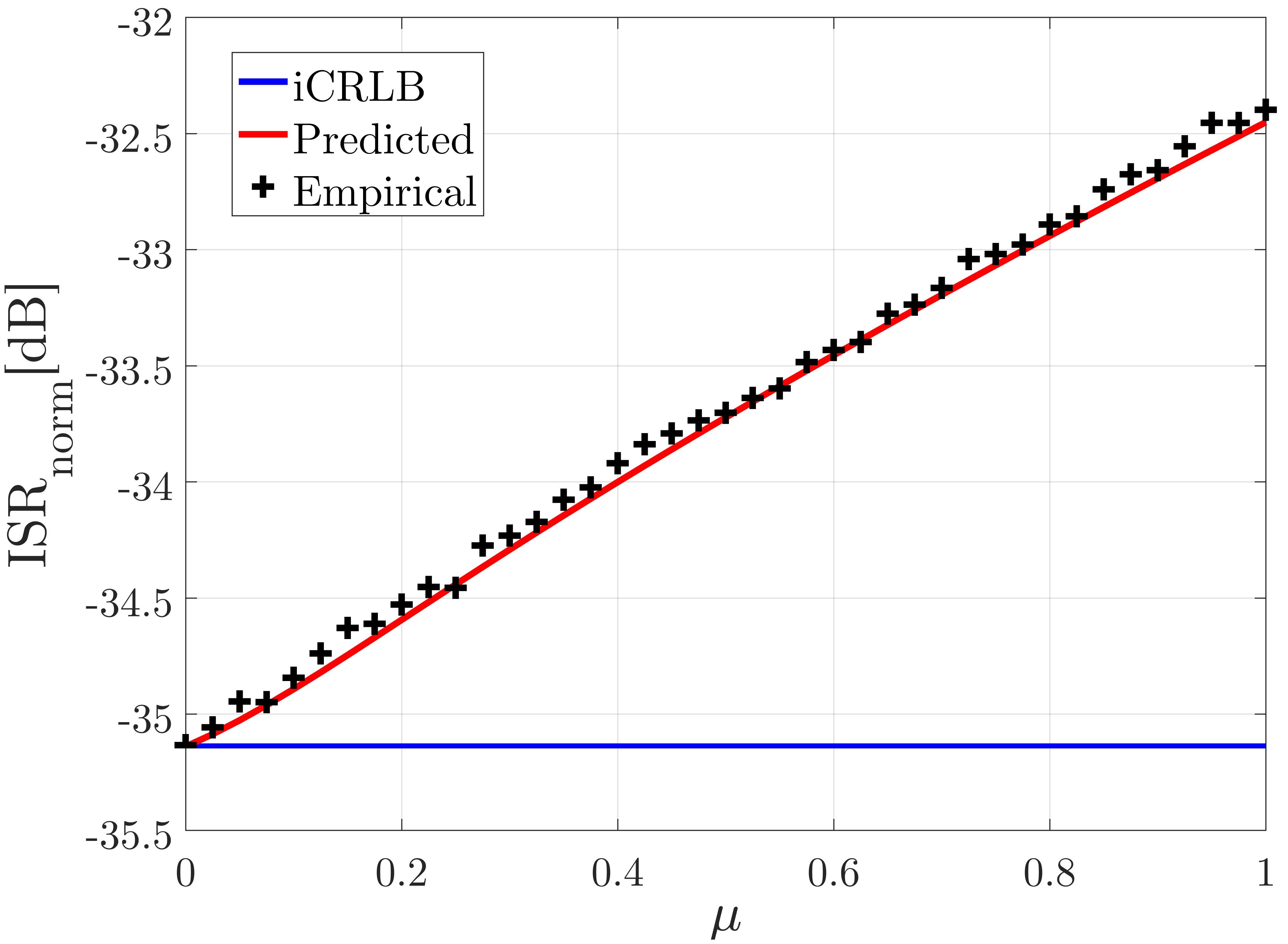}
	\caption{Analytically predicted and empirical total normalized ISR vs. $\mu$, representing the deviation from the model, i.e., the true zeros of the FIR filters. An excellent match is evident, where the largest difference between the \addra{predicted} and empirical values is $\sim0.1$[dB].\addra{ The empirical results were obtained by averaging $10^4$ independent trials.}}
	\label{fig1simul1}
\end{figure}
\begin{figure}
	\centering
	\includegraphics[width=0.5\textwidth]{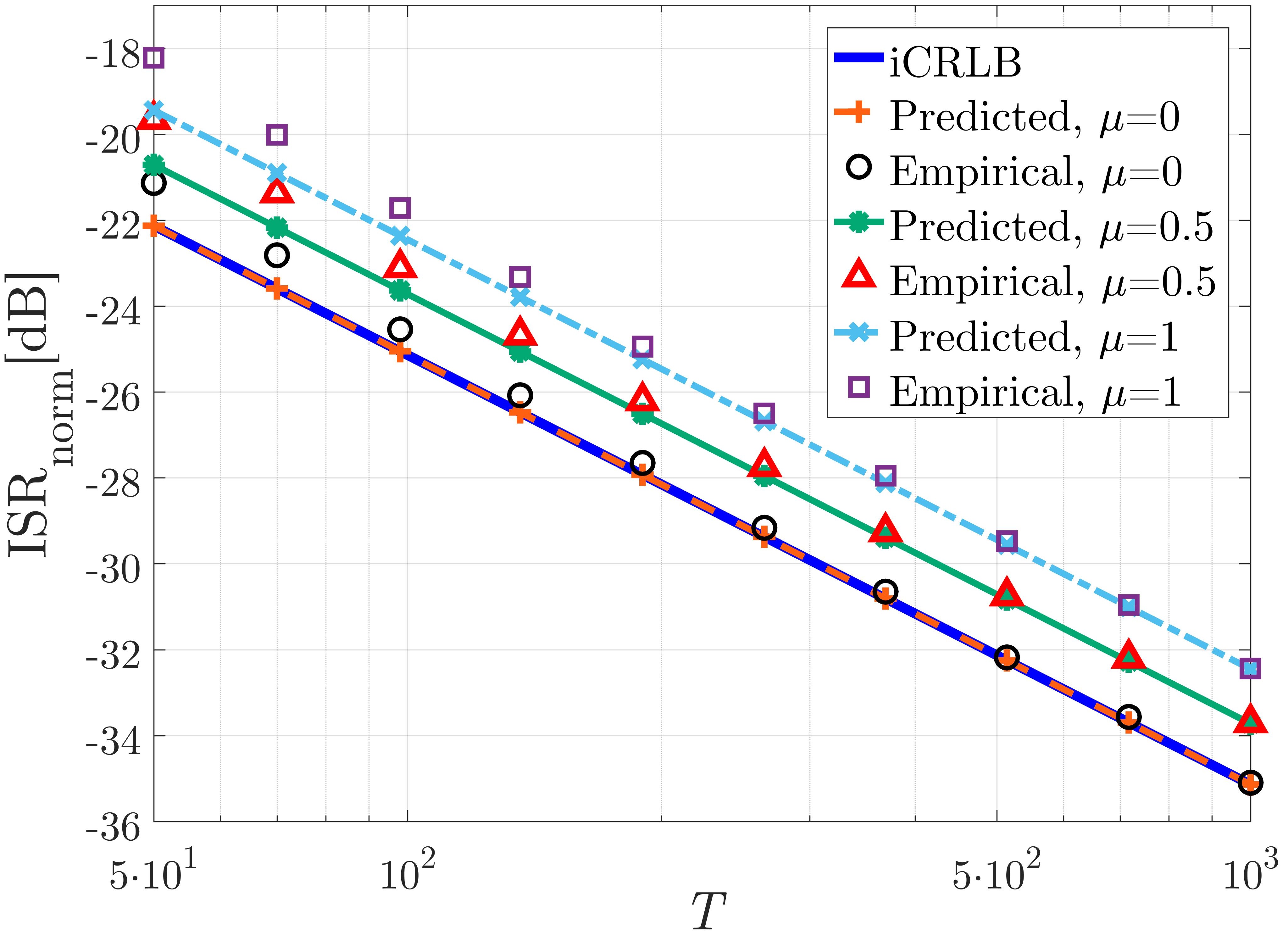}
	\caption{Analytically predicted and empirical total normalized ISR vs. $T$. It can be seen that the predicted ISR values, attained using the ``first-order" approximation based on the ``small-errors" assumption, are very close to the empirical values\delra{ obtained by the simulation.}\addra{, obtained by averaging $10^4$ independent trials.}} 
	\label{fig2simul1}
	\vspace{-0.3cm}
\end{figure}
\begin{figure*}[h]
	\includegraphics[width=\textwidth]{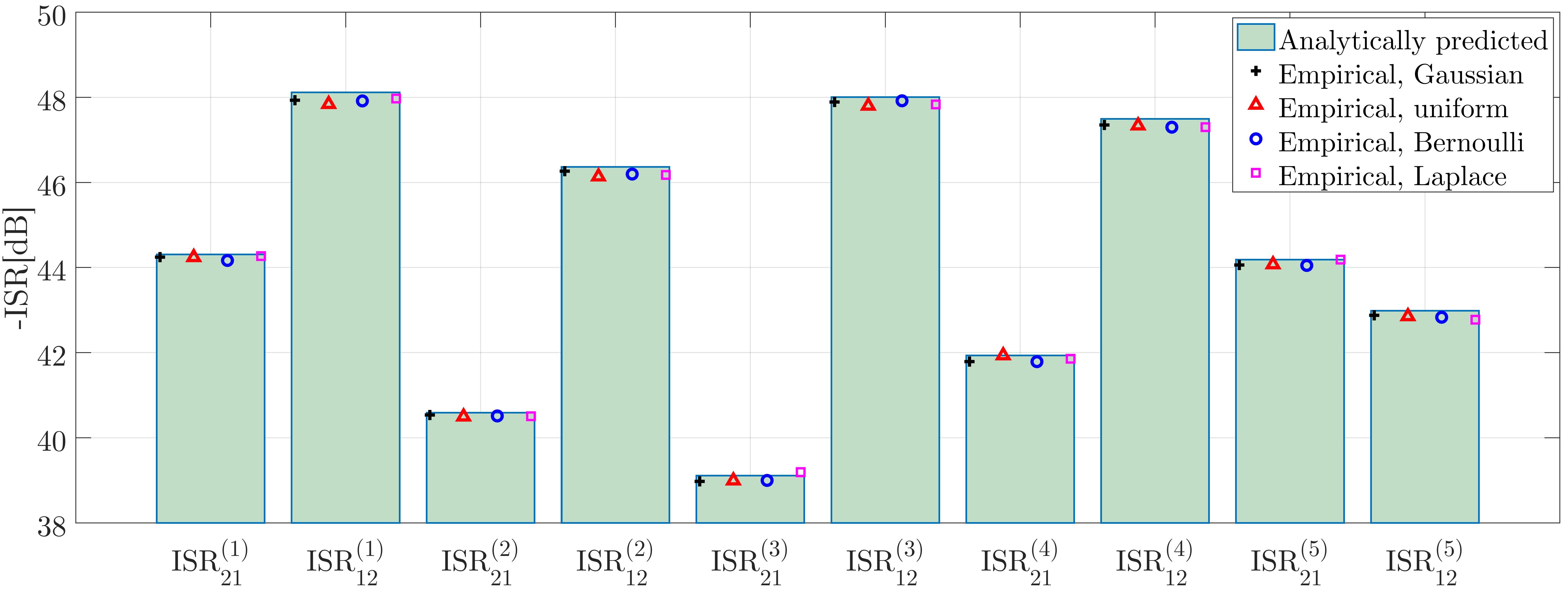}
	\caption{Analytically predicted and empirical (minus) ISR values for the Gaussian, uniform, Bernoulli and Laplace distributions. Clearly, the simulation results reflect a proper match between the theoretical predication and the empirical results \textit{regardless} of the distribution - the predicted ISR values are within $\sim0.2$[dB] from the empirical values\delra{.}\addra{, obtained by averaging $10^4$ independent trials.}}
	\label{fig2simul2}
\end{figure*}
In the first experiment we consider two datasets ($M=2$), each with three Gaussian sources ($K=3$). The white driving-noises $\left\{w_k^{(m)}[t]\right\}$ are all standard Gaussian processes, i.e., with zero-mean and unit variance. In this experiment, we assume that only some inaccurate versions $\wtC_k^{(m_1,m_2)}$ of the sources' true temporal covariance matrices  $\C_k^{(m_1,m_2)}$ are available, so the target-matrices $\hQ_k^{(m_1,m_2)}$ in \eqref{target_matrices_def} are computed using the inaccurate $\wtP_k^{(m_1,m_2)}$ matrices implied by the inaccurate $\wtC_k^{(m_1,m_2)}$ matrices via \eqref{SCM_def}, \eqref{invereseblockmatrices}. The covariance information is given, equivalently, in the form of inaccurate versions $\left\{\widetilde{h}_k^{(m_1,m_2)}[t]\right\}$ of the true FIR filters $\left\{h_k^{(m_1,m_2)}[t]\right\}$.

The true FIR filters (of length $L=10$) are generated in the following manner. For each FIR filter (i.e., for each triplet $\{k,m_1,m_2\}$) we draw one real-valued zero and four complex-valued zeros, all inside the unit-circle  (the remaining four zeros are the conjugate of the complex-valued ones, since the FIRs are real-valued). For each zero, the phase is drawn (independently) from $U(0,\pi)$ and the radius is drawn (also independently) from the distribution of the random variable $\exp\left(-au\right)$, where $a$ is a ``small" real-valued fixed parameter and $u\sim U(0,1)$ (the phase of the real-valued zero is fixed at zero). Denote these zeros as $\left\{z_{0,k}^{(m_1,m_2)}\right\}$. For each FIR we calculate its coefficients from its zeros, and then normalize its energy to agree with \eqref{FIR_energy}\addra{, and we set $\eta=1$}.

Now, to generate the ``inaccurate" version of each FIR we define the ``erroneous zeros", denoted by $z_{1,k}^{(m_1,m_2)}$, taking $z_{0,k}^{(m_1,m_2)}$ with their phases perturbed by a zero-mean Gaussian noise with variance $b^2$ and their radii attenuated by $c$, where $b,c$ are \delra{``small"}\addra{fixed} real-valued constants (such that the \addra{erroneous }zeros still lie inside the unit-circle). Now, we assume that the zeros of the inaccurate versions of the FIRs are given in the form of $\widetilde{z}_{k}^{(m_1,m_2)}=\mu z_{1,k}^{(m_1,m_2)}+(1-\mu)z_{0,k}^{(m_1,m_2)}$ (for each triplet $\{k,m_1,m_2\}$), where $0\leq \mu \leq 1$, from which we construct the inaccurate FIRs' coefficients (normalized to agree with \eqref{FIR_energy}), leading to the resulting inaccurate (``presumed") covariance matrices $\wtC_k^{(m_1,m_2)}$ and the ensuing target-matrices. For this experiment, we have used $a=2\addra{, b=0.1[\text{rad}]}$ and $\delra{b=}c=0.1$. As can be seen from Fig. \ref{fig1simul1}, which presents the total normalized ISR vs. $\mu$ (representing the ``amount" of estimation noise in the FIRs), there is an excellent match between the analytically predicted and the empirical results, so that the largest deviation is $\sim0.1$[dB] \footnote{This match was also verified for every ISR element independently}. As expected, when $\mu=0$ the iCRLB is attained, since the prior information on the model (i.e., the covariance matrices and the distribution of the sources) is accurate and the ML\addra{E}\delra{ estimate} is asymptotically efficient. However, as $\mu$ increases, the inaccuracy in the ``presumed" FIRs \delra{becomes more dominant, which in turn }leads to increasingly inaccurate versions of the presumed covariance matrices, thus the resulting ISR departs from its iCRLB, since the extended SeDJoCo solution for the noisy target-matrices is no longer the (exact) ML\addra{E}\delra{ estimate}, but only a QML\addra{E}\delra{ estimate}. Nevertheless, the performance is still accurately predicted by our analysis.

Fig. \ref{fig2simul1} presents the the total normalized ISR vs. $T$, the sample size, for three different values of $\mu$. Clearly, the ``first-order" approximation, assuming the validity of the ``small-errors" analysis, gives quite accurate predictions of the resulting ISRs even for very small sample sizes (e.g., for $T=50$ the empirical ISR values are within $\sim1$[dB] from the predicted values). Obviously, as $T$ grows, i.e., in the asymptotic regime, the approximation becomes even more accurate.
\vspace{-0.2cm}
\subsection{Sources' Distribution Mismodeling}\label{distributionmismodeling}
Our second experiment concerns the modeling, or more precisely, mismodeling of the sources' distributions. Indeed, due to the nature of the IVA problem, an inherent uncertainty is the one relating to the sources' distributions. In a semi-blind scenario, where prior knowledge is given or assumed, it is more realistic that only partial information on the statistical properties of the sources (e.g., SOS) is given rather than the entire distribution. Hence, in this experiment we assume perfect knowledge of the covariance matrices of the sources, but no information at all regarding the distributions of the sources. We consider five datasets ($M=5$), each with two sources ($K=2$). The white driving-noises $\left\{w_k^{(m)}[t]\right\}$ are drawn from four different distribution for each test case, where we examine the Gaussian, uniform, Bernoulli and Laplace distributions. Recall that for all cases, the white noises are zero-mean and unit variance, which uniquely defines all the parameters of all these distributions (aside from the Bernoulli distribution, which is traditionally defined for the values $0$ and $1$, and was defined here for the values $-1$ and $1$ with equal probability). The white driving-noises are then filtered by FIR filters of length $L=4$, whose coefficients are drawn from the standard Gaussian distribution. Notice that now we chose $L$ to be relatively small in order to keep the sources' marginal distributions away from the Gaussian distribution (as $L$ grows, the marginal distributions of the sources become ``more" Gaussian due to the central limit theorem). We also decrease the cross-correlation between corresponding sources from different datasets by setting $\eta=0.5$. Then, we obtain the extended SeDJoCo equations solution, which is a consistent estimate (according to the arguments presented in Section \ref{sec_ProblemFormulation}, Subsection \ref{properties_of_e_sedjoco}) but definitely not the ML\addra{E}\delra{ estimate} for non-Gaussian sources. Fig. \ref{fig2simul2} shows (for all the ISR elements, from all datasets) the analytically predicted values evaluated by \eqref{ISRintermsofCq}, and the empirical values for each test case obtained by simulation. It is readily seen that all predicted values well-approximate the true empirical values, to a precision of $\sim0.2$[dB], for each ISR element, and for all the distributions. We note in passing that the distributions for this experiment were chosen as representatives of both platykurtic (e.g., uniform) and leptokurtic (e.g., Laplace) distributions, which produce less and more extreme outliers, respectively, than does the Gaussian distribution. Furthermore, we emphasize that in this experiment the covariance matrices were assumed to be perfectly known only in order to isolate the sources' distributions mismodeling effect. Of course, our results are also valid for cases where both factors - inaccurate covariance matrices and sources' distributions mismodeling - are considered\addra{, as demonstrated in the following experiment}.
\vspace{-0.2cm}
\addra{\subsection{General Mismodeling and a Comparison to IVA-G-N}
	In our third experiment we consider a scenario \delra{with general mismodeling, }where both mismodeling factors, which were separately examined in the first two experiments, are introduced simultaneously. In addition, we compare the separation performance obtained by the Gaussian QMLE vs. Anderson {\it et al.}'s \cite{anderson2012joint} Newton updates for Gaussian IVA (IVA-G-N). We stress that the IVA-G-N is intended for separation of temporally independent identically distributed (i.i.d.) Gaussian sources, which is only a particular (possible) case for the Gaussian QML presumption. Nonetheless, it is our purpose in this work to show how temporal diversity of the sources can be exploited for better separation (even without precise knowledge of the actual covariance structures governing this diversity).
	We consider $M=8$ datasets, each with $K=5$ Gaussian mixture sources. Here, the $k$-th source of the $m$-th dataset is defined as
	\begin{multline}
		s_k^{(m)}[t] = I_k^{(m)}[t]\cdot s_{k,a}^{(m)}[t]+\left(1-I_k^{(m)}[t]\right)\cdot s_{k,b}^{(m)}[t] \\
		\forall k\in\{1,\ldots,K\}, \forall m\in\{1,\ldots,M\},
	\end{multline}
	where the ``switch indices" $\left\{I_k^{(m)}[t]\in\{0,1\}\right\}$ are all mutually statistically independent i.i.d. Bernoulli processes, with a common (known) probability of ``success" $\text{Pr}\left(I_k^{(m)}[t]=1\right)\triangleq p\in(0,1)$, statistically independent of the sets $\left\{s_{k,a}^{(m)}[t]\right\}$ and $\left\{s_{k,b}^{(m)}[t]\right\}$, which were generated (according to the general description given in the beginning of this section) by the mutually statistically independent sets of white Gaussian driving-noises $\left\{w_{k,a}^{(m)}[t]\right\}$ and $\left\{w_{k,b}^{(m)}[t]\right\}$, and the sets of FIR filters $\left\{h_{k,a}^{(m_1,m_2)}[t]\right\}$ and $\left\{h_{k,b}^{(m_1,m_2)}[t]\right\}$, respectively. This time, we set $L=10, \eta=0.1$ (yielding weaker correlations between sources from different sets than in the two previous experiments) and $\mu=0.5$. For the QMLE we used only noisy versions of the presumed covariance matrices ensuing from the erroneous versions of the FIRs $\left\{\widetilde{h}_{k,a}^{(m_1,m_2)}[t]\right\}$ and $\left\{\widetilde{h}_{k,b}^{(m_1,m_2)}[t]\right\}$ (given implicitly by the sets of their corresponding erroneous zeros $\left\{\widetilde{z}_{k,a}^{(m_1,m_2)}\right\}$ and $\left\{\widetilde{z}_{k,b}^{(m_1,m_2)}\right\}$, respectively). The sample size was set to $T=10^4$, so that in this scenario we also investigate the proposed algorithm in a considerably larger-scale estimation problem, including $MK^2=200$ unknown parameters and $MKT=400,000$ available samples.
	
	Fig. 4 presents the total normalized ISR vs. $p$, which is \textit{de facto} the Gaussian mixture characterizing parameter quantifying (in this case) the non-Gaussianity of the sources. At $p=0$ and $p=1$, the sources are (truly) Gaussian, such that the modeling error is only due to the inaccuracy in the presumed covariance matrices. However, as $p$ departs from the edges towards $p=0.5$, the sources' distributions depart from Gaussianity towards the ``most" non-Gaussian form of this parametric distribution - an equiprobable two Gaussians mixture. As expected, it can be seen that even with both mismodeling errors, the ``first-order" approximation gives an accurate prediction of the resulting ISR. Moreover, the Gaussian QMLE, attaining its lower bound, exhibits superiority over the IVA-G-N for every value of $p$ in this scenario, which demonstrates its ability to exploit temporal diversity even when the cross-correlations are relatively low.
	\begin{figure}
		\centering
		\includegraphics[width=0.5\textwidth]{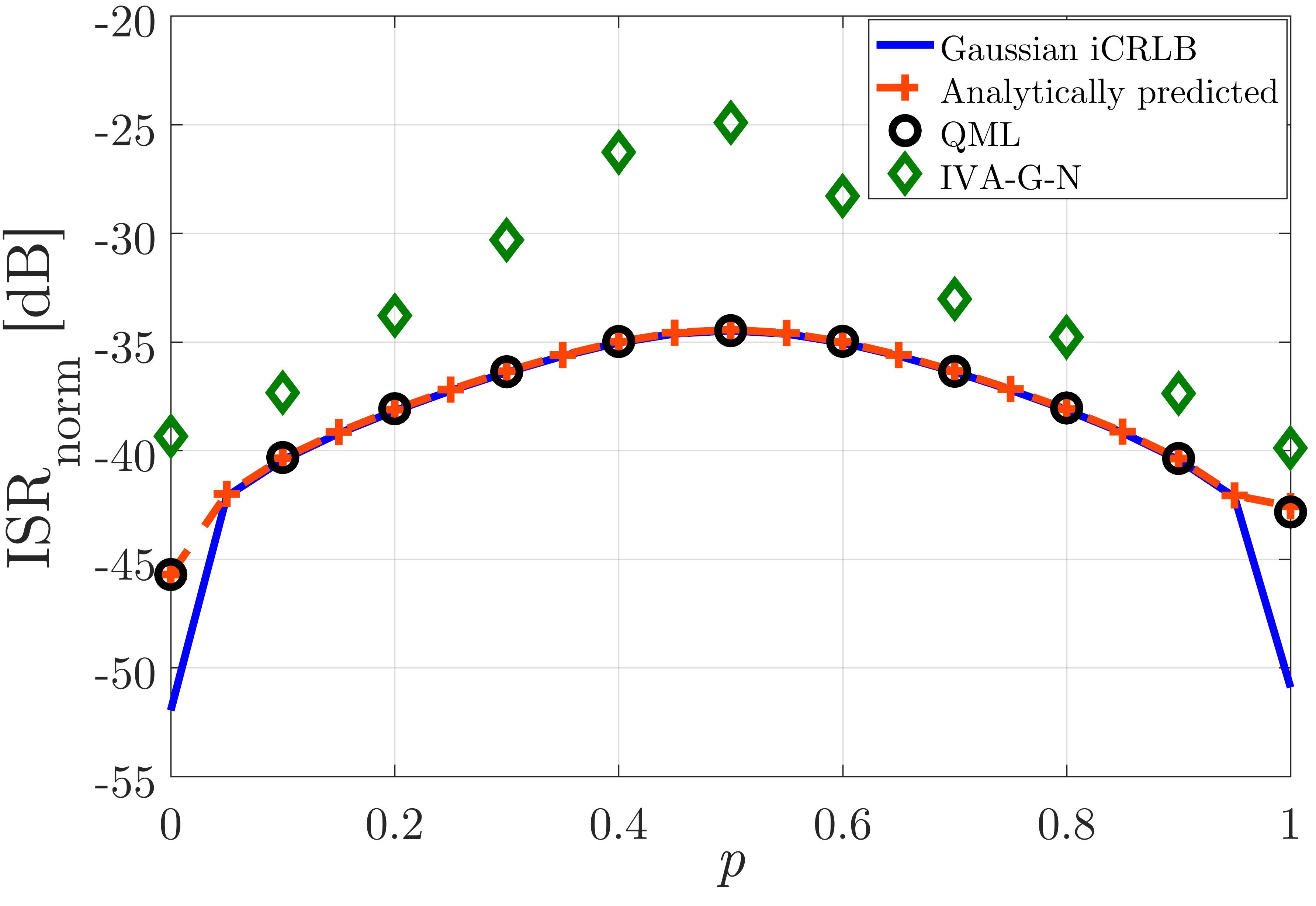}
		\caption{\addra{Analytically predicted and empirical total normalized ISR vs. $p$, representing deviation from Gaussianity. For a large sample size ($2000$ [sample/estimated parameter]), the approximated prediction is seen to be very accurate. Additionally, the Gaussian QMLE is superior to the IVA-G-N for every $p$. Empirical results were obtained by averaging $10^3$ independent trials.}}
		\label{fig1simul3}
		\vspace{-0.5cm}
	\end{figure}
	\captionsetup[figure]{labelfont={color=black}}
}

\section{Conclusion}
\label{sec_Conclusion} We considered a separation procedure for a semi-blind IVA problem using the extended SeDJoCo solution. We showed that (under mild conditions) this solution is consistent (in the sense of zero ISR when $T\rightarrow\infty$) for a general IVA problem, i.e., potentially fully blind. In addition, using the classical ``first-order" perturbation analysis (under the ``small-errors" assumption), we provided closed-form analytical expressions for the approximated predicted ISRs. \addra{Furthermore, we proved the ``universal" asymptotic ISR Theorem for the Gaussian QMLE and showed that the Gaussian iCRLB serves as an asymptotically attainable lower bound for the resulting ISR, regardless of the true sources' distributions, as an immediate implicit result of the Theorem. }Our analytical results were verified by simulations, demonstrating an excellent match with the obtained empirical outcomes (to a precision of $\sim0.2$[dB] for a sample size of $10^3$) for two types of partial prior knowledge - imperfect SOS and sources' distributions mismodeling.
\section{Acknowledgment}
\label{sec_acknowledgment} The authors gratefully acknowledge the financial support by the German-Israeli Foundation (GIF), grant number I-1282-406.10/2014. The first author also wishes to thank the Yitzhak and Chaya Weinstein Research Institute for Signal Processing for a fellowship.
\appendices
\addra{\section{Partial Differentiation of the Matrix $\F$}\label{appendix_RQ}
Differentiating each block $\F_m$ w.r.t. $Q_{k,ij}^{(m_1,m_2)}$ yields
\begin{align}\label{first_term_mth_matrix_appendix}
\begin{split}
&\frac{\partial\F_m}{\partial Q_{k,ij}^{(m_1,m_2)}} = \sum_{\tilde{k}=1}^{K}\sum_{\ell=1}^{M}\E_{\tilde{k}\tilde{k}}{\B^{(\ell)}\frac{\partial\Q_{\tilde{k}}^{(\ell,m)}}{\partial Q_{k,ij}^{(m_1,m_2)}}{\B^{(m)\rm{T}}}}=\\
&\delta_{mm_2}\E_{kk}\tb_i^{(m_1)}\tb_j^{(m_2)\rm{T}}+\\
&\left\{\begin{array}{ll}
\delta_{mm_1}\E_{kk}\tb_j^{(m_2)}\tb_i^{(m_1)\rm{T}} & \mbox{$m_1\neq m_2$}\\
\delta_{mm_2}\E_{kk}\tb_j^{(m_2)}\tb_i^{(m_1)\rm{T}} & \mbox{$m_1=m_2,i\neq j$}\\
0 & \mbox{$m_1=m_2,i=j$} \end{array} \right.,
\end{split}
\end{align}
where $\tb_i^{(m)}$ denotes the $i$-th column of $\B^{(m)}$. Then, in the same fashion, differentiating each block $\F_n$ w.r.t. $B_{pq}^{(m)}$ gives
\begin{align}\label{second_term_mth_matrix_appendix}
\begin{split}
&\frac{\partial\F_n}{\partial B_{pq}^{(m)}} = \sum_{k=1}^{K}\sum_{\ell=1}^{M}\E_{kk}{\frac{\partial}{\partial B_{pq}^{(m)}}}\left(\B^{(\ell)}\Q_{k}^{(\ell,n)}\B^{(n)\rm{T}}\right)=\\
&\sum_{k=1}^{K}\E_{kk}\Bigg[\delta_{mn}\left(\B^{(m)}\Q_k^{(m,m)}\E_{qp}+\E_{pq}\Q_k^{(m,m)}+\right.\\
&\sum_{\substack{\ell=1\\\ell\neq m}}^{M}\left.\B^{(\ell)}\Q_k^{(\ell,m)}\E_{qp}\right)+\left(1-\delta_{mn}\right)\left(\E_{pq}\Q_k^{(m,n)}\B^{(n)\rm{T}}\right)\Bigg]=\\
&\sum_{k=1}^{K}\E_{kk}\left(\E_{pq}\Q_k^{(m,n)}\B^{(n)\rm{T}}+\delta_{mn}\sum_{\ell=1}^{M}\B^{(\ell)}\Q_k^{(\ell,m)}\E_{qp}\right),\\
&\E_{pq}\Q_p^{(m,n)}\B^{(n)\rm{T}}+\delta_{mn}\sum_{k=1}^{K}\sum_{\ell=1}^{M}\E_{kk}\B^{(\ell)}\Q_k^{(\ell,m)}\E_{qp},
\end{split}
\end{align}
where we have used $\E_{kk}\E_{pq}=\delta_{kp}\E_{pq}$ and (\cite{petersen2008matrix})
\begin{align}
&\frac{\partial\mX\B\mX^{\rm{T}}}{\partial X_{ij}}=\mX\B\E_{ji}+\E_{ij}\B\mX^{\rm{T}}, \;\;\frac{\partial\mX\A}{\partial X_{ij}}=\E_{ij}\A. \nonumber
\end{align}
}
\vspace{-0.5cm}
\section{Computation of the Elements of $E\left[\huq\huq^{\rm{T}}\right]$} \label{appendix_a}
Let us examine all different cases of the expression
\begin{multline*}\label{secondtermofCq_appendix}
  \left(E\left[\huq\huq^{\rm{T}}\right]\right)_{\tilde{i}\tilde{j}}\triangleq \mLambda_{\tilde{i}\tilde{j}}=\\
  \frac{1}{T^2}E\left[\us_{i_1}^{(m_1)\rm{T}}\P_{k_1}^{(m_1,n_1)}\us_{j_1}^{(n_1)}\us_{i_2}^{(m_2)\rm{T}}\P_{k_2}^{(m_2,n_2)}\us_{j_2}^{(n_2)}\right].
\end{multline*}
\underline{Case I:} $i_1=j_1, i_2=j_2, i_1\neq i_2$
\begin{equation*}
    \mLambda_{\tilde{i}\tilde{j}}=\frac{1}{T^2}\Tr\left(\C_{i_1}^{(n_1,m_1)}\P_{k_1}^{(m_1,n_1)}\right)\Tr\left(\C_{i_2}^{(n_2,m_2)}\P_{k_2}^{(m_2,n_2)}\right),
\end{equation*}
where we have used
\begin{align*}\label{appendixAeq1}
  &(\text{i})\;\;E\left[\us_{i_1}^{(m_1)\rm{T}}\P_{k_1}^{(m_1,n_1)}\us_{i_1}^{(n_1)}\us_{i_2}^{(m_2)\rm{T}}\P_{k_2}^{(m_2,n_2)}\us_{i_2}^{(n_2)}\right]=\\
  &\;\;\;\;\;\;E\left[\us_{i_1}^{(m_1)\rm{T}}\P_{k_1}^{(m_1,n_1)}\us_{i_1}^{(n_1)}\right]E\left[\us_{i_2}^{(m_2)\rm{T}}\P_{k_2}^{(m_2,n_2)}\us_{i_2}^{(n_2)}\right];\\
  &(\text{ii})\;\;\us_{i_1}^{(m_1)\rm{T}}\P_{k_1}^{(m_1,n_1)}\us_{j_1}^{(n_1)}=\Tr\left(\us_{j_1}^{(n_1)}\us_{i_1}^{(m_1)\rm{T}}\P_{k_1}^{(m_1,n_1)}\right); \text{and}\\
  &(\text{iii})\;\; E\left[\Tr\left(\cdot\right)\right]=\Tr\left(E\left[\cdot\right]\right).
\end{align*}
\underline{Case II:} $i_1=j_2, i_2=j_1, i_1\neq j_1$
\begin{equation*}
    \mLambda_{\tilde{i}\tilde{j}}=\frac{1}{T^2}\Tr\left(\C_{i_1}^{(n_2,m_1)}\P_{k_1}^{(m_1,n_1)}\C_{i_2}^{(n_1,m_2)}\P_{k_2}^{(m_2,n_2)}\right),
\end{equation*}
where we have used $(\text{i})$-$(\text{iii})$ and also
\begin{align*}\label{appendixAeq2}
  &(\text{iv})\;\;\us_{i_1}^{(m_1)\rm{T}}\P_{k_1}^{(m_1,n_1)}\us_{j_1}^{(n_1)}=\us_{j_1}^{(n_1)\rm{T}}\P_{k_1}^{(n_1,m_1)}\us_{i_1}^{(m_1)}; \text{and}\\
  &(\text{v})\;\;\Tr\left(\A\B\C\D\right)=\Tr\left(\D\A\B\C\right), \A,\B,\C,\D\in\Rset^{K\times K}.
\end{align*}
\underline{Case III:} $i_1=i_2, j_1=j_2, i_1\neq j_1$
\begin{equation*}
    \mLambda_{\tilde{i}\tilde{j}}=\frac{1}{T^2}\Tr\left(\C_{i_1}^{(m_2,m_1)}\P_{k_1}^{(m_1,n_1)}\C_{j_1}^{(n_1,n_2)}\P_{k_2}^{(n_2,m_2)}\right),
\end{equation*}
where we have used $(\text{i})$-$(\text{v})$.

\noindent\underline{Case IV:} $i_1=j_1=i_2=j_2$
\begin{align*}
    \mLambda_{\tilde{i}\tilde{j}}=\frac{1}{T^2}\bigg [ &\Tr\left(\C_{i_1}^{(n_1,m_1)}\P_{k_1}^{(m_1,n_1)}\right)\Tr\left(\C_{i_1}^{(n_2,m_2)}\P_{k_2}^{(m_2,n_2)}\right)+\\
    &\Tr\left(\C_{i_1}^{(n_2,m_1)}\P_{k_1}^{(m_1,n_1)}\C_{i_1}^{(n_1,m_2)}\P_{k_2}^{(m_2,n_2)}\right)+\\
    &\Tr\left(\C_{i_1}^{(m_2,m_1)}\P_{k_1}^{(m_1,n_1)}\C_{j_1}^{(n_1,n_2)}\P_{k_2}^{(n_2,m_2)}\right)\bigg ],
\end{align*}
where we have used $(\text{i})$-$(\text{v})$ and Isserlis' theorem \cite{isserlis1918formula}, regarding the computation of higher-order moments of the multivariate Gaussian distribution, and in particular,
\begin{align*}
  &E\left[s_{i_1}^{(m_1)}[t_1]s_{i_1}^{(n_1)}[t_2]s_{i_1}^{(m_2)}[t_3]s_{i_1}^{(n_2)}[t_4]\right]=\\
  &E\left[s_{i_1}^{(m_1)}[t_1]s_{i_1}^{(n_1)}[t_2]\right]E\left[s_{i_1}^{(m_2)}[t_3]s_{i_1}^{(n_2)}[t_4]\right]+\\
  &E\left[s_{i_1}^{(m_1)}[t_1]s_{i_1}^{(m_2)}[t_3]\right]E\left[s_{i_1}^{(n_1)}[t_2]s_{i_1}^{(n_2)}[t_4]\right]+\\
  &E\left[s_{i_1}^{(m_1)}[t_1]s_{i_1}^{(n_2)}[t_4]\right]E\left[s_{i_1}^{(n_1)}[t_2]s_{i_1}^{(m_2)}[t_3]\right]=\\
  &C_{i_1\addra{,t_1t_2}}^{(m_1,n_1)}\delra{[t_1,t_2]}C_{i_1\addra{,t_3t_4}}^{(m_2,n_2)}\delra{[t_3,t_4]}+C_{i_1\addra{,t_1t_3}}^{(m_1,m_2)}\delra{[t_1,t_3]}C_{i_1\addra{,t_2t_4}}^{(n_1,n_2)}\delra{[t_2,t_4]}+C_{i_1\addra{,t_1t_4}}^{(m_1,n_2)}\delra{[t_1,t_4]}C_{i_1\addra{,t_2t_3}}^{(n_1,m_2)}\delra{[t_2,t_3]}.
\end{align*}
\underline{Case V:} otherwise
\begin{equation*}
    \mLambda_{\tilde{i}\tilde{j}}=0.
\end{equation*}
\vspace{-0.75cm}
\section{\delra{Computation of the Elements of $\G$}\addra{Proof of Theorem \ref{theorem1}}} \label{appendix_b}
\addra{We shall now show that the solution of \eqref{derivative_solution} for every $\utheta_r, 1\leq r \leq M_q$,
\begin{align}\label{zerogradientatii}
  &\left.\frac{dB_{pq}^{(m)}}{d\widehat{Q}_{k,ij}^{(m_1,n_1)}}\right|_{\huq=\uq}=0, p\neq q, \\
  \quad\quad\forall p,q,i,j,k\in&\{1,\ldots,K\}, \forall m,m_1,n_1\in\{1,\ldots,M\}, \nonumber
\end{align}
which zeros out all the elements of $\C_{\hat{q}_{(I)}}$ (in \eqref{ISRintermsofCq}, when the sum is written explicitly) that depend on the fourth-order statistics of the sources\addra{, thereby proving the Theorem}.

By virtue of the ISR-equivariance property of the (Q)ML\addra{E}\delra{ estimate} of the demixing matrices, it suffices to consider}\delra{We are interested in} the solution of \eqref{derivative_solution} for the case $\bA=\I_{KM}=\bB$ (where $\bB\triangleq\bA^{-1}$), evaluated at $\huq=\uq$. Notice that in this case the terms in equation \eqref{first_term_matrix} boil down to
\begin{equation}\label{YforAIQQH}
\left.\frac{\partial\F}{\partial \widehat{Q}_{k,ij}^{(m_1,n_1)}}\right|_{\bB=\I_{KM}}=\delta_{ij}\delta_{ik}\E_{kk}\cdot\left\{\begin{array}{ll}
\delta_{mm_1}&\mbox{$m_1=n_1$}\\
\delta_{mm_1}+\delta_{mn_1}&\mbox{$m_1\neq n_1$}\end{array}.\right.
\end{equation}
First, notice that $i\neq j$ or $i\neq k$ implies $\Y=\O$, so for all these cases we have that $\utheta=\uo$, i.e.,
\begin{multline}\label{thetaforineqj}
  \quad\quad\left.\frac{dB_{pq}^{(m)}}{d\widehat{Q}_{k,ij}^{(m_1,n_1)}}\right|_{\huq=\uq}=0, \forall m,m_1,n_1\in\{1,\ldots,M\},\\
  \forall i,j,k,p,q\in\{1,\ldots,K\}: (i\neq j) \cup (i\neq k).
\end{multline}
Now, let us examine the cases where $i=j=k$. Obviously, there are two such cases: $m_1=n_1$ and $m_1\neq n_1$. We start with the case where $m_1=n_1$, and we assert that for this case the solution of \eqref{derivative_solution} is of the form
\begin{equation}\label{solution_of_theta_appendix}
  \utheta = \sum_{\ell=1}^{M}\alpha_{\ell}\ue_{i+(i-1)K+(\ell-1)K^2},
\end{equation}
where $\{\alpha_{\ell}\}_{\ell=1}^M$ are some coefficients (to be defined below in \eqref{computationofcoefs3}).
To show this, we begin by computing $\uy$ for this case. By \eqref{first_term_matrix} and \eqref{YforAIQQH}, it follows that
\begin{equation}\label{Yfor1stcase}
\begin{aligned}
& \Y=\left[\O \; \cdots \; \O \underbrace{\E_{ii}}_{m_1\text{-th matrix}} \O \; \cdots \; \O\right],\\
&\quad\quad\quad\quad\quad\quad\quad\Leftrightarrow \\
& \uy=\text{vec}\left(\Y\right)=\ue_{i+(i-1)K+(m_1-1)K^2}.
\end{aligned}
\end{equation}
Then, recall that the columns of the matrix $\H$ are the vectors $\left\{\uh_{[m,p,q]}\right\}$, defined in \eqref{total_derivative_vec_form}. These vectors are the concatenated column vectors of the matrices $\left\{\H_{[m,p,q]}\right\}$, respectively, hence
\begin{equation}\label{pinoutvectorz}
  \H\ue_{i+(i-1)K+(\ell-1)K^2}=\uh_{[i,i,\ell]}=\text{vec}\left(\H_{[i,i,\ell]}\right).
\end{equation}
Evaluating \eqref{second_term_mth_matrix} at $\huq=\uq$, and substituting $\bB=\I_{KM}$, we have that
\begin{align}\label{delFdelBappendix}
\begin{split}
  \left.\frac{\partial\F_m}{\partial B_{ii}^{(\ell)}}\right|_{\huq=\uq} &= \sum_{k=1}^{K}\E_{kk}\left(\E_{ii}\Q_k^{(\ell,m)}+\delta_{m\ell}\sum_{\tilde{m}=1}^{M}\Q_k^{(\tilde{m},\ell)}\E_{ii}\right)\\
  &=\E_{ii}\Q_i^{(\ell,m)}+\delta_{m\ell}\sum_{k=1}^{K}\sum_{\tilde{m}=1}^{M}\E_{kk}\Q_k^{(\tilde{m},\ell)}\E_{ii}\\
  &=Q_{i,ii}^{(\ell,m)}\E_{ii}+\delta_{m\ell}\sum_{\tilde{m}=1}^{M}Q_{i,ii}^{(\tilde{m},\ell)}\E_{kk}\\
  &=\left[\sum_{\tilde{m}=1}^{M}Q_{i,ii}^{(\tilde{m},\ell)}\left(\delta_{m\ell}+\delta_{m\tilde{m}}\right)\right]\E_{ii}\triangleq \beta(\ell,m)\E_{ii},
\end{split}
\end{align}
where we have used the fact that when $\bB=\I_{KM}=\bA$, all the matrices $\left\{\Q_k^{(m,\ell)}\right\}$ are diagonal, and that $Q_{i,ii}^{(\ell,m)}=Q_{i,ii}^{(m,\ell)}$. Using the definition in \eqref{definition_of_Y_and_H}, it follows that
\begin{align}\label{pinoutvectorz2}
\begin{split}
  \H\utheta&=\sum_{\ell=1}^{M}\alpha_{\ell}\H\ue_{i+(i-1)K+(\ell-1)K^2}\\
  &=\sum_{\ell=1}^{M}\sum_{\tilde{\ell}=1}^{M}\alpha_{\ell}\beta(\ell,\tilde{\ell})\ue_{i+(i-1)K+(\tilde{\ell}-1)K^2}\\
  &=\sum_{\tilde{\ell}=1}^{M}\left(\sum_{\ell=1}^{M}\alpha_{\ell}\beta(\ell,\tilde{\ell})\right)\ue_{i+(i-1)K+(\tilde{\ell}-1)K^2}\\
  &\triangleq \sum_{\tilde{\ell}=1}^{M}\gamma_{\tilde{\ell}}\ue_{i+(i-1)K+(\tilde{\ell}-1)K^2}.
\end{split}
\end{align}
Now, notice that
\begin{align}\label{yappearsinthesolution}
\begin{split}
  \H\utheta&=\sum_{\tilde{\ell}=1}^{M}\gamma_{\tilde{\ell}}\ue_{i+(i-1)K+(\tilde{\ell}-1)K^2}\\
  &=\gamma_{m_1}\ue_{i+(i-1)K+(m_1-1)K^2}+\sum_{\substack{\tilde{\ell}=1 \\ \tilde{\ell}\neq m_1}}^{M}\gamma_{\tilde{\ell}}\ue_{i+(i-1)K+(\tilde{\ell}-1)K^2}\\
  &= \gamma_{m_1}\uy+\sum_{\substack{\tilde{\ell}=1 \\ \tilde{\ell}\neq m_1}}^{M}\gamma_{\tilde{\ell}}\ue_{i+(i-1)K+(\tilde{\ell}-1)K^2},
\end{split}
\end{align}
so that if $\gamma_{\ell}=-\delta_{m_1\ell}$, \eqref{solution_of_theta_appendix} is a solution of \eqref{derivative_solution}. Indeed, this can be easily obtained by the following linear system of equations
\begin{equation}\label{computationofcoefs}
  \sum_{\ell=1}^{M}\alpha_{\ell}\beta(\ell,\tilde{\ell})=-\delta_{m_1\tilde{\ell}}, \forall \tilde{\ell}\in\{1,\ldots,M\},
\end{equation}
or written equivalently in matrix form
\begin{equation}\label{computationofcoefs2}
  \mBeta\ualpha=-\ue_{m_1},
\end{equation}
where the $(i,j)$-th element of the matrix $\mBeta\in\Rset^{M\times M}$ is $\beta(i,j)$, and the $i$-th element of $\ualpha\in\Rset^{M\times1}$ is $\alpha_i$. Finally, with
\begin{equation}\label{computationofcoefs3}
  \ualpha=-\mBeta^{-1}\ue_{m_1},
\end{equation}
we have that
\begin{equation}\label{finaleqinproof}
  \H\utheta=-\uy,
\end{equation}
and we conclude that \eqref{solution_of_theta_appendix} is the solution of \eqref{derivative_solution}.

We shall now address the last case where $m_1\neq n_1$ (recall that in this case also $i=j=k$). The only difference from the previous case is the value of $\uy$, since now by \eqref{YforAIQQH} (w.l.o.g. assume $n_1>m_1$)
\begin{equation}\label{Yfor2ndcase}
\begin{aligned}
& \quad\quad\Y=\left[\O \; \cdots \; \O \underbrace{\E_{ii}}_{m_1\text{-th matrix}} \O \; \cdots \; \O \underbrace{\E_{ii}}_{n_1\text{-th matrix}} \O \; \cdots \; \O\right],\\
&\quad\quad\quad\quad\quad\quad\quad\quad\quad\quad\quad\quad\Leftrightarrow \\
& \uy=\text{vec}\left(\Y\right)=\ue_{i+(i-1)K+(m_1-1)K^2}+\ue_{i+(i-1)K+(n_1-1)K^2}.
\end{aligned}
\end{equation}
Applying exactly the same machinery used in the previous case, we select the coefficients $\left\{\alpha_{\ell}\right\}_{\ell=1}^M$ for the solution vector \eqref{solution_of_theta_appendix} such that
\begin{equation}\label{computationofcoefsB}
  \sum_{\ell=1}^{M}\alpha_{\ell}\beta(\ell,\tilde{\ell})=-\left(\delta_{m_1\tilde{\ell}}+\delta_{n_1\tilde{\ell}}\right), \forall \tilde{\ell}\in\{1,\ldots,M\},
\end{equation}
or written equivalently in matrix form
\begin{equation}\label{computationofcoefs2B}
  \mBeta\ualpha=-\left(\ue_{m_1}+\ue_{n_1}\right),
\end{equation}
so that \eqref{solution_of_theta_appendix} with the coefficients
\begin{equation}\label{computationofcoefs3B}
  \ualpha=-\mBeta^{-1}\left(\ue_{m_1}+\ue_{n_1}\right),
\end{equation}
is a solution of \eqref{derivative_solution}. With this, we have presented the closed-form solution for the gradient of each element in the solution matrices of the extended SeDJoCo w.r.t. each element of the target-matrices for the particular case where $\bA=\I_{KM}=\bB$, evaluated at $\huq=\uq$.

\addra{Now, }\delra{R}\addra{r}ecall the definition of $\utheta$ in \eqref{definitionfothetaandH}, and notice th\addra{at for}\delra{e following important property of} the aforementioned solution
\begin{align}\label{zerogradientatiiappendix}
  &\left.\frac{dB_{pq}^{(m)}}{d\widehat{Q}_{k,ij}^{(m_1,n_1)}}\right|_{\huq=\uq}=0, p\neq q, \\
  \quad\quad\forall p,q,i,j,k\in&\{1,\ldots,K\}, \forall m,m_1,n_1\in\{1,\ldots,M\}, \nonumber
\end{align}
which means that asymptotically, the resulting\delra{ (approximated)} ISR from the separating solution of the extended SeDJoCo depends only on the \delra{SOS of the}sources\addra{' SOS}, as seen from equation \eqref{ISRintermsofCq} when written explicitly as a sum of \delra{multiplications}\addra{products} between the elements of the gradient \addra{vector $\ug_{ij}^{(m)}$ }and \addra{the} covariance matrix $\C_{\hat{q}_{(I)}}$.}

\bibliography{Bibfile}
\bibliographystyle{unsrt}

\end{document}